\documentclass[a4paper, 12pt]{article}

\usepackage[sort&compress]{natbib}
\bibpunct{(}{)}{;}{a}{}{,} 
\RequirePackage[colorlinks,citecolor=blue,urlcolor=blue]{hyperref}

\usepackage{amsthm, amsmath, amssymb, mathrsfs, multirow, url}
\usepackage{graphicx} 
\usepackage{ifthen} 
\usepackage{amsfonts}
\usepackage[usenames]{color}
\usepackage{fullpage}
\usepackage{tikz}
\usepackage{float}
\usepackage{bbm}
\usepackage{booktabs}
\usepackage{subcaption}
\allowdisplaybreaks

\theoremstyle{plain} 
\newtheorem{thm}{Theorem}
\newtheorem{cor}{Corollary}
\newtheorem{prop}{Proposition}
\newtheorem{lem}{Lemma}

\theoremstyle{definition}

\newtheorem*{asmp}{Regularity Conditions}

\theoremstyle{remark} 
\newtheorem{ex}{Example}

\newcommand{\ra}[1]{\renewcommand{\arraystretch}{#1}}

\newcommand{\prob}{\mathsf{P}}

\newcommand{\unif}{{\sf Unif}}
\newcommand{\nm}{{\sf N}}
\newcommand{\expo}{{\sf Exp}}
\newcommand{\gam}{{\sf Gamma}}

\newcommand{\chisq}{{\sf ChiSq}}

\newcommand{\RR}{\mathbb{R}}

\newcommand{\ZZ}{\mathbb{Z}}

\newcommand{\TT}{\mathbb{T}}

\newcommand{\lPi}{\underline{\Pi}}
\newcommand{\uPi}{\overline{\Pi}}

\begin{document}

\title{\textbf{Divide-and-conquer with finite sample sizes: valid and efficient possibilistic inference}}
\author{Emily C. Hector\textsuperscript{$\star$}, Leonardo Cella\textsuperscript{$\dagger$} and Ryan Martin\textsuperscript{$\star$}\hspace{.2cm}\\
\textsuperscript{$\star$}Department of Statistics, North Carolina State University\\
\textsuperscript{$\dagger$}Department of Statistical Sciences, Wake Forest University}
\date{}
\maketitle
 
\bigskip
\begin{abstract}
Divide-and-conquer methods use large-sample approximations to provide frequentist guarantees when each block of data is both small enough to facilitate efficient computation {\em and} large enough to support approximately valid inferences. When the overall sample size is small or moderate, likely no suitable division of the data meets both requirements, hence the resulting inference lacks validity guarantees. We propose a new approach, couched in the {\em inferential model} framework, that is fully conditional in a Bayesian sense and provably valid in a frequentist sense. The main insight is that existing divide-and-conquer approaches make use of a Gaussianity assumption twice: first in the construction of an estimator, and second in the approximation to its sampling distribution. Our proposal is to retain the first Gaussianity assumption, using a Gaussian working likelihood, but to replace the second with a validification step that uses the sampling distributions of the block summaries determined by the posited model. This latter step, a type of probability-to-possibility transform, is key to the reliability guarantees enjoyed by our approach, which are uniquely general in the divide-and-conquer literature.  In addition to finite-sample validity guarantees, our proposed approach is also asymptotically efficient like the other divide-and-conquer solutions available in the literature.  Our computational strategy leverages state-of-the-art black-box likelihood emulators. We demonstrate our method's performance via simulations and highlight its flexibility with an analysis of median PM\textsubscript{2.5} in Maryborough, Queensland, during the 2023 Australian bushfire season.
\end{abstract}

\noindent%

\maketitle 

\section{Introduction} \label{s:intro}

Divide-and-conquer techniques have emerged as powerful tools for big data analysis when data sets' sample sizes are so large that, when divided across a number of central processing units (CPUs), the subsets remain sufficiently large and rich to yield approximately valid inference. A more challenging but still common setting is when estimation alone is computationally demanding enough that analyzing the entire dataset is infeasible, yet the sample size is not large enough to guarantee (approximately) valid inference. In this paper, we develop divide-and-conquer methodology for computationally challenging optimization problems in this latter setting, couched in the {\em inferential models} (IMs) framework of \cite{martin2013inferential, martin2015inferential} and \cite{martin2019false}.  The IM framework's two defining features are:
\begin{itemize}
\item it is fully conditional in a Bayesian sense, meaning that it assigns data-dependent degrees of belief to all assertions about the unknown parameter; and
\vspace{-2mm}
\item it is reliable or {\em valid} in a frequentist sense, meaning that assigning high degrees of belief to false assertions about the unknowns is a low probability event. 
\end{itemize}
Details are given in Section~\ref{ss:oracle-IM}, but an important consequence of validity is that inferential procedures derived from the IM output, such as confidence regions and hypothesis tests, are provably calibrated at any desired level, independent of sample size.

Broadly speaking, the divide-and-conquer framework takes a data set of sample size $n$ that is computationally prohibitive to analyze in its entirety and splits it into $B$ blocks of roughly equal sample size to be analyzed in parallel across a distributed computing system. In almost all divide-and-conquer methods (see the review in Section \ref{s:dac-background}), block-specific sample sizes are assumed large and asymptotic Gaussianity of the estimating functions and/or estimators in the $B$ blocks is used to derive a combined estimator and approximate its distribution. The combined estimator, which we hereafter call the ``large-$n$'' estimator, typically takes the form of a weighted average, with weights given by the inverse asymptotic variance of block-specific estimators. When the sample size within each block is not particularly large, the asymptotic distribution of the large-$n$ estimator is not appropriately calibrated. In practice, this often leads to inflated Type-I errors and an inflated tendency to incorrectly conclude that an effect of inferential interest is statistically non-zero. We illustrate further in Section~\ref{s:simulations} the deleterious consequences of using inappropriately calibrated inferences. 

One motivating example is the \emph{g-and-k} family of distributions \citep{haynes1997robustness, rayner2002numerical}, defined by the quantile function
\begin{align*}
Q(u)&= \mu + \sigma z_u \left( 1+c \, \frac{1-e^{-gz_u}}{1+e^{-gz_u}} \right) (1+z_u^2)^{k}, \quad 0 \leq u \leq 1,
\end{align*}
where $\mu \in \mathbb{R}$ is a location parameter, $\sigma > 0$ is a scale parameter, $g \in \mathbb{R}$ measures skewness and $k >-1/2$ measures kurtosis, $z_u = \Phi^{-1}(u)$ is the $u^\text{th}$ standard Gaussian quantile, and $c$ is a constant corresponding to the value of ``overall symmetry''.  If $g<0$, then the distribution is skewed to the left; likewise, $g>0$ indicates skewness to the right. The \emph{g-and-k} family can capture a broad class of distributional shapes. The family can also represent shorter tails than the Gaussian when $k<0$. Due to their flexibility, this family has been used to model complex financial and climate data, among others. Finding the maximum likelihood estimator in such models is difficult because each log-likelihood evaluation requires solving the inverse problem $y_i=Q(u_i)$ for each observation $y_i$, $=1, \ldots, n$. The optimization is therefore not only expensive but also numerically challenging. Furthermore, when $n$ is not too large (e.g., $n=200$), computationally prohibitive likelihood evaluations are needed to ensure validity of inference. As the distribution is easy to sample from, it is frequently used to illustrate the use of approximate Bayesian computation \citep[e.g.,][]{fearnhead2012constructing}, which can be slow and difficult to tune and has no frequentist calibration guarantees. We reduce the computational burden and guarantee valid inference through divide-and-conquer in the IM framework. 

Our main contribution in this paper is the development of a divide-and-conquer IM framework that, in addition to offering Bayesian-like fully conditional uncertainty quantification, is frequentistly valid in finite-samples and equivalent to the asymptotically efficient full-data IM in large samples. In other words, our proposed IM gains important statistical validity guarantees, compared to existing approaches that focus solely on asymptotic validity, at no (asymptotic) statistical efficiency loss---one may have their cake and eat it too.  Further, we develop a new computationally efficient tool for evaluation of our proposed valid divide-and-conquer IM based on a black-box likelihood \emph{emulator}, along with other more basic strategies to speed up IM computations. 

Section~\ref{s:background} gives an overview of the divide-and-conquer framework and IMs. Section~\ref{ss:IM-connection} lays the groundwork by looking at two extreme versions of a divide-and-conquer IM solution: one that is ``optimal'' but practically out of reach and one that is incredibly simple but only asymptotically valid.  Building on this experience, Section~\ref{s:proposal} describes our new approach for divide-and-conquer inference that achieves both finite-sample validity and asymptotic efficiency compared to an oracle solution that has the necessary computational resources to handle the full data. Section~\ref{s:simulations} demonstrates the performance of our proposed solution in several numerical examples, including an analysis of median PM\textsubscript{2.5} in Maryborough, Queensland, during the 2023 Australian bushfire season.  All code to reproduce the results is available at \url{https://github.com/ehector/IMdac}. 

\section{Background and notation} 
\label{s:background}

\subsection{Divide-and-conquer methods}
\label{s:dac-background}

Divide-and-conquer methods for analyzing massive distributed data emerged directly from \cite{Glass-1976}'s meta-analysis. The primary task in a divide-and-conquer framework is to obtain unified inference across $B$ independent blocks of data that is both computationally and statistically efficient. The size of each block must be small enough that it can be quickly analyzed, but large enough that estimates are approximately valid, leading to a fundamental tension between computational and statistical efficiency.  Data summaries from each block are typically used to reduce the communication and computation costs, so the main challenge and focus of divide-and-conquer methods is in developing statistically and computationally efficient rules for combining these summaries.

To fix notation, let $Z^n = (Z_1,\ldots,Z_n)$ consist of $n$ independent observables having joint distribution $\prob_\Theta$, depending on an uncertain true parameter $\Theta$ taking values in the parameter space $\TT \subseteq \RR^p$; the dependence of $\prob_\Theta$ on $n$ is omitted.  The individual $Z_i$'s may represent pairs $(X_i, Y_i)$ of predictor and response variables, as in an observational study, or the predictor variables may be fixed constants, as in a designed experiment. In any case, the full data $Z^n$ is randomly divided into $B$ blocks of sizes $n_1,\ldots,n_B$ of comparable sizes; we denote these blocks by $Z^{(b)}$ for $b=1,\ldots,B$.  

As indicated above, it is often the case that only summaries of the data blocks are available.  We will assume that, for each $b=1,\ldots,B$, the summary $S_b = (n_b, \hat\theta_{Z^{(b)}}, J_{Z^{(b)}})$ of $Z^{(b)}$ includes the block sample size $n_b$, the $p$-dimensional maximum likelihood estimator $\hat\theta_{Z^{(b)}} \in \TT$, and the $p \times p$ observed Fisher information matrix $J_{Z^{(b)}}$. The aggregation of the $B$ summaries is $S^n = (S_{1},\ldots,S_{B})$, and the ``total information'' $J_{S^n} = \sum_{b=1}^B J_{Z^{(b)}}$.  Finally, we use lowercase letters $z^n$, $z^{(b)}$, $s_b$, $s^n$, $J_{s^n}$, etc.~to denote realizations of the random variables $Z^n$, $Z^{(b)}$, $S_b$, $S^n$, $J_{S^n}$, etc. 

Divide-and-conquer methods have spanned multiple topics from kernel ridge regression \citep{zhang2015divide}, high-dimensional sparse regression \citep{Lee-etal-2017, lin2019race}, variable screening \citep{diao2024distributed} and empirical likelihood \citep{zhou2023distributed} to modeling matrices \citep{mackey2015distributed, nezakati2023unbalanced}, high-dimensional correlated data \citep{hector2020doubly, hector2021distributed, hector2022joint} and spatial fields \citep{lee2023scalable, hector2024distributed, hector2025distributed}. See \cite{chen2021divide, zhouetal2023distributed, Hector-etal-2024} for recent reviews. The divide-and-conquer methodology we focus on is termed ``one-shot'' because each block of data is analyzed only once (we therefore omit literature on one-step updates and surrogate likelihoods). Most one-shot approaches rely on (weighted) averaging, where the combined estimator is the (weighted) average of the study estimators \citep[e.g.][]{lin2011aggregated, shi2018massive, hector2023parallel}. 

Related to our work, inspired by Fisher's fiducial inference \citep{Fisher-1935, fisher1956statistical} and \cite{Efron-1993}'s confidence distribution, \cite{Singh-Xie-Strawderman-2005, Xie-Singh-Strawderman-2011, Liu-Liu-Xie-2014, Liu-Liu-Xie-2015, Yang-Liu-Wang-Xie-2016, Michael-Thornton-Xie-Tian, Tang-Zhou-Song-2020} proposed to combine inferences across studies using a frequentist confidence distribution. In this body of work, the confidence distribution is a sample-dependent function that encodes all confidence levels of a parameter. In contrast to the possibility contours obtained through the IM framework, introduced in Section \ref{ss:oracle-IM} below, the primary focus of this framework is on controlling the behaviour of the confidence distribution at a point null hypothesis.  Care must be taken, however, since applying the familiar probability calculus---that is, integration---to confidence distributions for broader uncertainty quantification creates risks \citep[e.g.,][]{fraser.cd.discuss, fraser2011} and, in particular, false confidence \citep{Ryansatellite}. 

\subsection{Inferential models}
\label{ss:oracle-IM}

As stated briefly in Section~\ref{s:intro}, the inferential model (IM) framework offers data-driven quantification of uncertainty about unknown parameters in statistical models, among other things.  This uncertainty quantification is designed to be fully conditional in a Bayesian sense and provably reliable in a frequentist sense.  Achieving both the Bayesian and frequentist goals simultaneously requires something beyond the textbook probability and statistical theory.  IMs' specific novelty is that its uncertainty quantification is couched in the language of {\em imprecise probability theory} or, more specifically, {\em possibility theory} \citep[e.g.,][]{dubois.prade.book, dubois2006possibility}.  While possibility theory may be unfamiliar to the reader, it is easy to explain, which we do now.  

In a single sentence, possibility theory is probability theory with integration replaced by optimization.  Start with a function $\pi: \TT \to [0,1]$ with the property $\sup_{\theta \in \TT} \pi(\theta) = 1$.  This function is called a {\em possibility contour} or simply a {\em contour}.  The ``supremum-equals-one'' condition parallels the ``integral-equals-one'' familiar normalization condition for probability density functions.  Then the contour $\pi$ determines a {\em possibility measure} $\uPi$ via optimization: 
\[ \uPi(H) = \sup_{\theta \in H} \pi(\theta), \quad H \subseteq \TT. \]
We postpone offering an interpretation of $\uPi(H)$ for now but, roughly speaking, $\uPi(H)$ can be interpreted as an upper probability, or an upper bound on a range of candidate subjective probabilities.  There is a corresponding lower probability $\lPi$, also determined by the contour $\pi$, but we will not need this in what follows.  Since we are not using the lower probability, we will drop the bar notation and simply write $\Pi$ for $\uPi$.  

The statistical inference problem involves data $Z^n$ from distribution $\prob_\Theta$, where $\Theta \in \TT$ is unknown or uncertain.  The IM framework maps the observed data $z^n$ to a possibility contour $\pi_{z^n}$ and corresponding possibility measure $\Pi_{z^n}$ supported on $\TT$.  We interpret $\Pi_{z^n}(H)$ as a data-driven measure of how possible the assertion ``$\Theta \in H$'' is, given the observed data $z^n$. Simply having the mathematical properties of a possibility measure is not enough to make the IM's uncertainty quantification meaningful.  This meaningfulness comes from frequentist-style calibration properties, the strongest and most fundamental of which is {\em strong validity}, i.e., 
\begin{equation}\label{eq:IMvalidity}
\prob_{\Theta}\{ \pi_{Z^n}(\Theta) \leq \alpha\} \leq \alpha, \quad  \text{for all $\alpha \in [0,1]$}.
\end{equation}
An immediate consequence is an easier-to-interpret property called {\em validity},
\[ \sup_{\theta \in H} \prob_\theta\{ \Pi_{Z^n}(H) \leq \alpha \} \leq \alpha, \quad \alpha \in [0,1], \quad H \subseteq \TT. \]
Roughly, validity implies that the IM assigning small possibility values to a true hypothesis is a rare event.  Therefore, as the familiar inductive logic goes, if we evaluate $\Pi_{z^n}(H)$ in our application and find that this number is small, then inferring $H^c$ would be warranted.  See \cite{cella2023possibility} for further details on interpretation. In addition, strong validity implies that procedures derived from the IM output achieve the familiar frequentist error rate control guarantees.  In particular, the set $C_\alpha(z^n) = \{\theta \in \TT: \pi_{z^n}(\theta) > \alpha\}$ is a genuine $100(1-\alpha)$\% confidence set for $\Theta$ in the sense that it has frequentist coverage probability at least $1-\alpha$, i.e., $\inf_\theta \prob_\theta\{ C_\alpha(Z^n) \ni \theta\} \geq 1-\alpha$.

Following \citet{martin2022valid,martinpp2,martinpp3}, one constructs a valid possibilistic IM in two steps: {\em ranking} and {\em validification}.  The ranking step involves the specification of a ranking function $R(z^n,\theta)$ that ranks the parameter value in terms of how compatible it is with the data $z^n$, with higher values of the ranking function meaning greater compatibility.  A very natural choice of the ranking function is the {\em relative likelihood}
\[ R(z^n, \theta) = L_{z^n}(\theta) / L_{z^n}(\hat\theta_{z^n}), \quad \theta \in \TT, \]
where $L_{z^n}(\theta)$ is the likelihood function based on $z^n$ and, as above, $\hat\theta_{z^n}$ is the maximum likelihood estimate based on $z^n$.  Then the IM's possibility contour is obtained through the validification step, which amounts to doing the following probability calculation:
\begin{align}    
\pi_{z^n}(\theta) = \prob_{\theta} \{ R(Z^n, \theta) \leq R (z^n, \theta)\}, \quad \theta \in \TT.
\label{e:oracle-IM}
\end{align}
This is a version of the so-called probability-to-possibility transform \citep[e.g.,][]{hoseandhass2020,hose2021universal} applied to the relative likelihood.  If evaluation of the likelihood and the maximum likelihood estimator are computationally expensive, then the IM contour defined above may be out of reach in practice. One of this paper's main constributions is a set of analytic and efficient computational strategies to approximate a possibilistic IM contour like that in \eqref{e:oracle-IM} in such cases; see Section \ref{ss:IM-connection}.

Although the cases we consider in this paper assume a parametric model for the data at hand, the ranking--validification construction can also be applied to distribution-free problems; see, e.g., \citet{cella.martin.imrisk}, \citet[][Sec.~6]{martinpp3}, and \cite{CELLA2024}. Additionally, while the relative likelihood is a natural choice for ranking in parametric problems, this is not the only option and, in fact, the validification step in \eqref{eq:IMvalidity} can be performed with any suitable ranking function $R$. This flexibility is relevant in Section~\ref{s:proposal}, where we introduce a new, strongly valid IM for divide-and-conquer inference.

\section{Towards a divide-and-conquer IM}
\label{ss:IM-connection}

\subsection{The ultima Thule}
\label{ss:thule}

The IMs approach described in Section~\ref{ss:oracle-IM} uses the whole data set $z^n$ to draw valid and efficient possibilistic inferences.  The aforementioned ranking and validification steps require two things: the relative likelihood (which implicitly depends on the maximum likelihood estimator) and its distribution. In settings where evaluating the likelihood is computationally expensive, obtaining the maximum likelihood estimator and computing the relative likelihood over a sufficiently dense grid of candidate values for $\Theta$ becomes prohibitive. If the relative likelihood can be computed, the validification step can be performed analytically when its distribution is available in closed-form. Otherwise, expensive computations are again needed to evaluate its distribution empirically.

Divide-and-conquer analysis aims to bypass simultaneous and potentially expensive computations with the whole data $z^n$ by combining cheaper, block-specific inferences on $\Theta$ based on the blocked data $z^{(b)}$, $b=1,\ldots,B$. To evaluate the contour in \eqref{e:oracle-IM} in a divide-and-conquer framework, a first idea might be to try reconstructing this full-data contour using only summary statistics $s_b$ from $z^{(b)}$. This is achievable when the relative likelihood $R(z^n, \theta)$ depends on the data $z^n$ only through $s^n$, the aggregate of the block summary statistics.  Below, we present two examples where this holds, aiming to build intuition about the IM construction and the complexity of the problem at hand.  Even for these simple cases where computation is virtually free, reconstructing the full-data IM solution from the blocked data is quite challenging, hence the need for a different approach.  Our proposed solution in Section~\ref{s:proposal} works very well even when the likelihood functions are computationally very expensive; see the examples in Section~\ref{s:simulations}.

\begin{ex}
\label{ss:gaussian}
Let $Z^n$ consist of $n$ iid samples from $\nm(\Theta, \tau^2)$, where $\tau^2 > 0$ is known.  Split the data into blocks $Z^{(b)}$ of size $n_b$, for $b=1,\ldots,B$. Then $\hat\theta_{Z^{(b)}}$ is the within-block sample mean, which is Gaussian with mean $\Theta$ and variance $J_{Z^{(b)}}^{-1} = n_b^{-1} \tau^2$, which does not depend on $Z^{(b)}$.
The block-specific relative likelihood is $R(z^{(b)}, \theta) = \exp\{- J_{z^{(b)}} (\theta - \hat\theta_{z^{(b)}})^2/2\}.$ The corresponding IM contour for each block $b=1,\ldots,B$ is
\begin{align*}
\pi_{z^{(k)}}(\theta) & = \prob_{\theta} \{ R(Z^{(b)}, \theta) \leq R(z^{(b)},\theta)\} \\
& = 1 - F_1\{J_{z^{(b)}} (\theta - \hat\theta_{z^{(b)}})^2 \} = 1 - F_1\{ n_b (\theta - \hat\theta_{z^{(b)}})^2 / \tau^2 \}, \quad \theta \in \RR, 
\end{align*}
with $F_1$ the $\chisq(1)$ distribution function.  The full-data IM contour is defined similarly: 
\begin{align}
\pi_{z^n}(\theta) &  = 1 - F_1\{ J_{s^n}(\theta - \hat\theta_{z^n})^2 \} = 1 - F_1\{ n (\theta - \bar z)^2 / \tau^2 \}, \quad \theta \in \RR, \label{e:gaussian-IM-oracle}
\end{align}
with $\hat\theta_{z^n}=\bar z$ and $J_{s^n} = n/\tau^2$ representing the full-data sample mean and Fisher information, respectively.  The key point is that the full-data IM contour can be obtained from the block-specific summary statistics alone.  For illustration, the left panel of Figure~\ref{f:gaussian-gamma-example} plots $\pi_{z^{(b)}}$, $b=1, 2, 3$ and $\pi_{z^n}$ for the case of $\Theta=0$, $n_b=5b$, and $\tau^2=2^2$. As expected, the full-data IM contour is tighter than the block-specific contours. 
\end{ex}

\begin{figure}[t]
\centering
\begin{subfigure}[t]{0.49\textwidth}
\includegraphics[width=\textwidth]{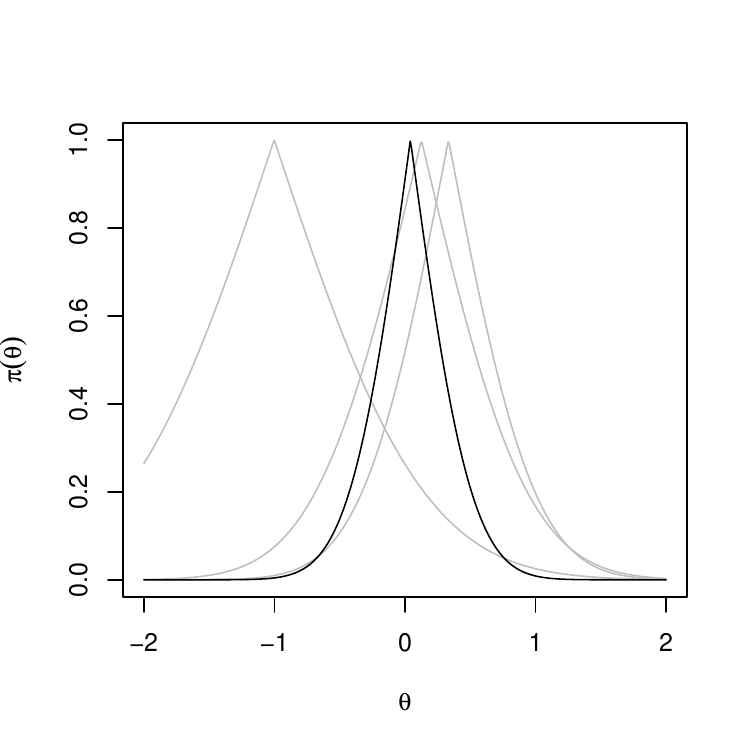}
\end{subfigure}\hfill
\begin{subfigure}[t]{0.49\textwidth}
\includegraphics[width=\textwidth]{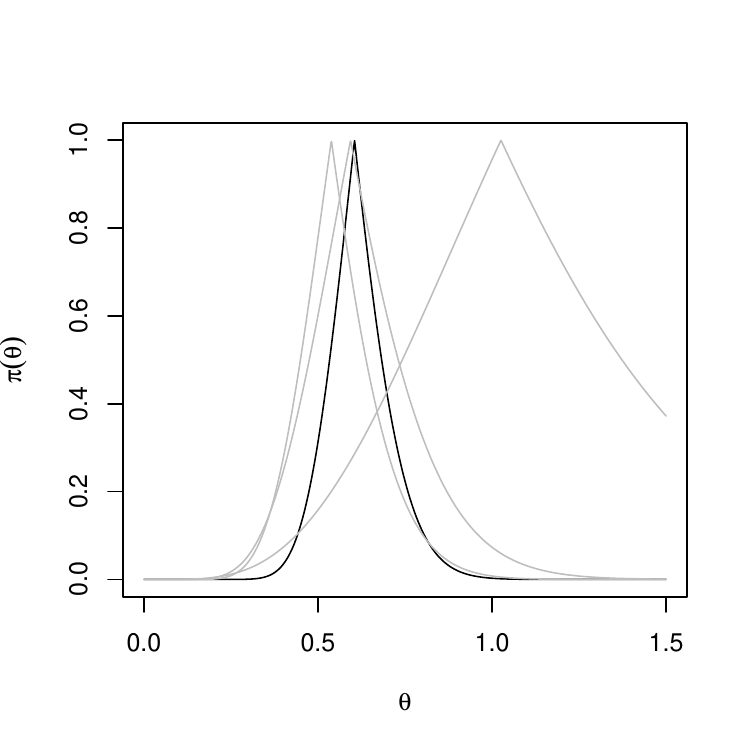}
\end{subfigure}
\caption{Full data (black) and block-specific (grey) possibility contours.  The left panel shows the Gaussian example and the right panel the Exponential example.} \label{f:gaussian-gamma-example}
\end{figure}

\begin{ex}
\label{ss:expo}
Let $Z^n$ consist of $n$ iid Exponential samples with unknown rate $\Theta > 0$.  Split the data into $B$ blocks of size $n_1,\ldots,n_B$, then the block-specific relative likelihood is 
\[ R(z^{(b)}, \theta) = ( \theta \hat{\theta}_{z^{(b)}}^{-1} )^{n_b} \exp\{ n_b(1-\theta \hat{\theta}_{z^{(b)}}^{-1}) \}, \quad \theta > 0, \quad b=1,\ldots,B, \]
where $\hat\theta_{z^{(b)}}$ is the block-$b$ sample mean's reciprocal.  The corresponding IM contour is 
\begin{equation}
\label{eq:likeli.ratio.exp}
\pi_{z^{(k)}}(\theta) = \prob_{\theta} \{ R(Z^{(k)}, \theta) \leq R(z^{(k)}, \theta) \}, \quad \theta > 0, \quad k=1,\ldots,K, 
\end{equation}
but there is no simple closed-form expression for this like there was in the previous Gaussian case.  There is, however, a not-so-simple closed-form expression for the contour---see Appendix~\ref{p:gamma-IM-proof}---and standard software can be readily used to approximate it. For illustration, the right panel of Figure~\ref{f:gaussian-gamma-example} plots $\pi_{z^{(b)}}$, $b=1,2,3$ and $\pi_{z^n}$ in an example from the Exponential model with $\Theta=0.5$ and $n_b=5b$. The efficiency improvement of the full data IM compared to the individual IMs is evident, as expected.
\end{ex}

\subsection{A large-sample divide-and-conquer IM} 
\label{s:asymptotic-rule}

The full data maximum likelihood estimator and the relative likelihood generally cannot be expressed as a function of the summary statistics alone; the Gaussian case is one exception. So, the classic approach to divide-and-conquer inference just \emph{assumes} that the block-specific maximum likelihood estimators are Gaussian and mimics the derivation in Example~\ref{ss:gaussian} above to obtain \citep[e.g.,][]{Hector-etal-2024, hedges1983combining}, 
\begin{align}
\check{\theta}_{S^n} = J_{S^n}^{-1} \sum_{b=1}^B J_{Z^{(b)}} \hat\theta_{Z^{(b)}},
\label{e:meta-def}
\end{align}
which is a weighted average of the block-specific maximum likelihood estimators. If the $\hat\theta_{Z^{(b)}}$'s are exactly independent Gaussian with mean $\Theta$, then $\check{\theta}_{S^n}$ is the best linear unbiased estimator of $\Theta$ and its covariance matrix is $J_{S^n}^{-1}$---which actually does not depend on the summary statistics $S^n$ in this Gaussian case.  

The property above that holds exactly in the Gaussian case holds more generally, at least approximately, when the sample size is large \citep[e.g.,][]{hedges1981distribution}.  That is, if all of the block sample sizes are sufficiently large, e.g., as in equation \eqref{eq:n.condition} below, then $\check\theta_{S^n}$ is approximately Gaussian with mean $\Theta$ and its variance can be consistently estimated by $J_{S^n}^{-1}$, which is of order $n^{-1}$.  Moreover, $\check\theta_{S^n}$ is asymptotically efficient in the sense that its asymptotic covariance matrix agrees with that of the maximum likelihood estimator based on the full sample $Z^n$; see Appendix~\ref{a:large.sample} and \cite{hedges1981distribution}.  

Next, we develop an analogue of this classical result that leans on the asymptotic Gaussianity of the block-specific maximum likelihood estimators to construct an asymptotically valid and efficient IM for $\Theta$ based on the summaries in $S^n$ alone.  The construction is very simple---we literally mimic the formulas in the exactly Gaussian case. The large-$n$ divide-and-conquer IM's possibility contour is
\begin{equation}
\label{eq:pi.naive}
\pi_{s^n}^\infty(\theta) = 1 - F_p\bigl\{ (\check\theta_{s^n} - \theta)^\top J_{s^n} (\check\theta_{s^n} - \theta) \bigr\}, \quad \theta \in \TT, 
\end{equation}
where $F_p$ is the $\chisq(p)$ distribution function. This IM construction depends on the full data $z^n$ only through the summary statistics in $s^n$, whereas the relative likelihood based on the full data $z^n$ cannot be evaluated with the summaries alone, so we have obviously given something up.  As the superscript ``$\infty$'' indicates, what we sacrifice for the extreme simplicity of this solution is that it is only finite-sample valid.  Remarkably, the asymptotically valid inferences obtained from this large-$n$ divide-and-conquer IM are efficient, as they asymptotically align closely with the full-data IM inferences, which are themselves asymptotically optimal \citep{martin2025asymptotic}. 

\begin{thm}
\label{thm:naive-efficiency}
Under the standard regularity conditions  stated in Appendix~\ref{a:large.sample}, if all the block sample sizes are increasing to infinity at the same rate, in the sense that 
\begin{equation}
\label{eq:n.condition}
n^{-1} n_b \to a_b >0, \quad n \to \infty, \quad b=1,\ldots,B, \quad \text{where} \; \textstyle\sum_{b=1}^B a_b = 1, 
\end{equation}
then the two possibilistic IM contours, $\pi_{s^n}^\infty$ and $\pi_{Z^n}$, merge locally uniformly, i.e., 
\begin{align*}
\sup_{u \in C} \bigl| \pi_{s^n}^\infty(\hat\theta_{Z^n} + J_{Z^n}^{-1/2} u) - \pi_{Z^n}(\hat\theta_{Z^n} + J_{Z^n}^{-1/2}u) \bigr| \to 0, \quad \text{in $\prob_{\Theta}$-probability, $n \to \infty$},
\end{align*}
where $C \subset \RR^p$ is an arbitrary compact set. 
\end{thm}

This is a strong approximation: the large-$n$ divide-and-conquer IM possibility contour is close to that of the full data IM, uniformly over $\theta$ values in a $n^{-1/2}$-neighborhood of $\Theta$. Figure~\ref{f:meta-IM}(a) illustrates this by plotting, in red, the large-$n$ divide-and-conquer IM contour from \eqref{eq:pi.naive} in the Exponential example of Section~\ref{ss:expo}, with $J_{Z^{(b)}} = n_i \hat{\theta}_{Z^{(b)}}^{-2}$. As before, the black curve corresponds to the full data IM contour, while the gray curves correspond to the individual IMs contours from the $B=3$ blocks. The large-$n$ divide-and-conquer IM is impressively efficient, but recall that its validity is guaranteed only asymptotically. To confirm this, Figure~\ref{f:meta-IM}(b) plots the (empirical) distribution function of $\pi_{Z^n}(\Theta)$ and $\pi_{s^n}^\infty(\Theta)$ where we simulate $10{,}000$ datasets from the Exponential distribution with rate parameter $\Theta = 0.5$  and evaluate the contours at $\Theta=0.5$ in each Monte Carlo replicate. 
We observe that the large-$n$ divide-and-conquer IM is (slightly) stochastically smaller than $\unif(0,1)$, indicating that it is not valid in the present setting with $n_b=5b$. To illustrate a consequence of this invalidity, we compute the confidence intervals $\{\theta: \pi_{s^n}^\infty(\theta) >\alpha\}$ for each of the $10{,}000$ Monte Carlo replicates. The empirical coverage probability, reported in Table~\ref{t:gamma-CP}, is the proportion of the $10{,}000$ computed intervals that contain the true value $\Theta=0.5$. As a preview of the results to come, Table~\ref{t:gamma-CP} also presents the empirical coverage probability of confidence intervals based on our proposed finite-sample valid IM, which will be introduced in Section~\ref{s:proposal}, with contour $\pi_{s^n}^\forall$.  
The empirical coverage probabilities of confidence intervals based on the large-$n$ divide-and-conquer contour fall below the nominal level, indicating inflated type-I error rates. The valid contour, on the other hand, tracks the nominal level closely (within margins of Monte Carlo error) across $\alpha$.

\begin{figure}[t]
\centering
\begin{subfigure}[t]{0.49\textwidth}
\includegraphics[width=\textwidth]{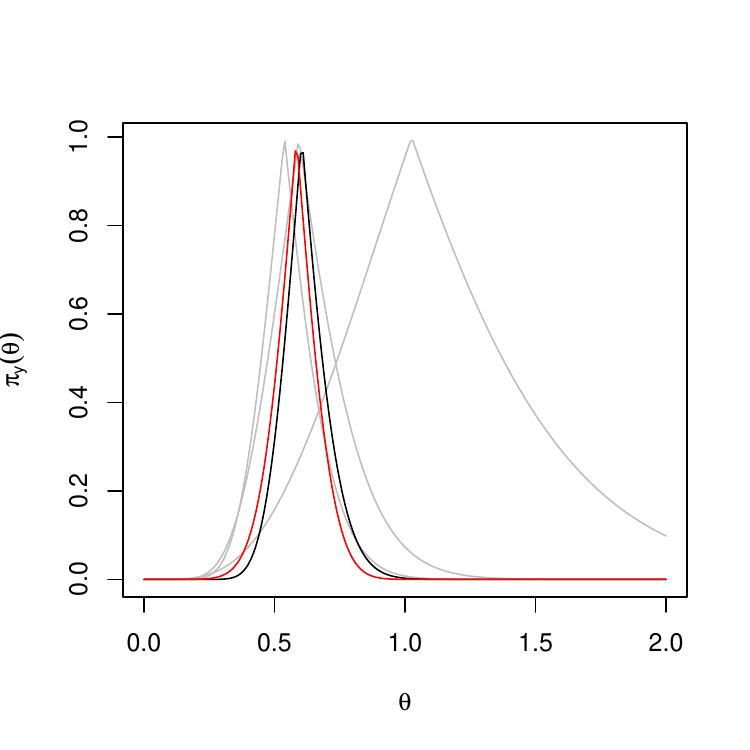}
\end{subfigure}\hfill
\begin{subfigure}[t]{0.49\textwidth}
\includegraphics[width=\textwidth]{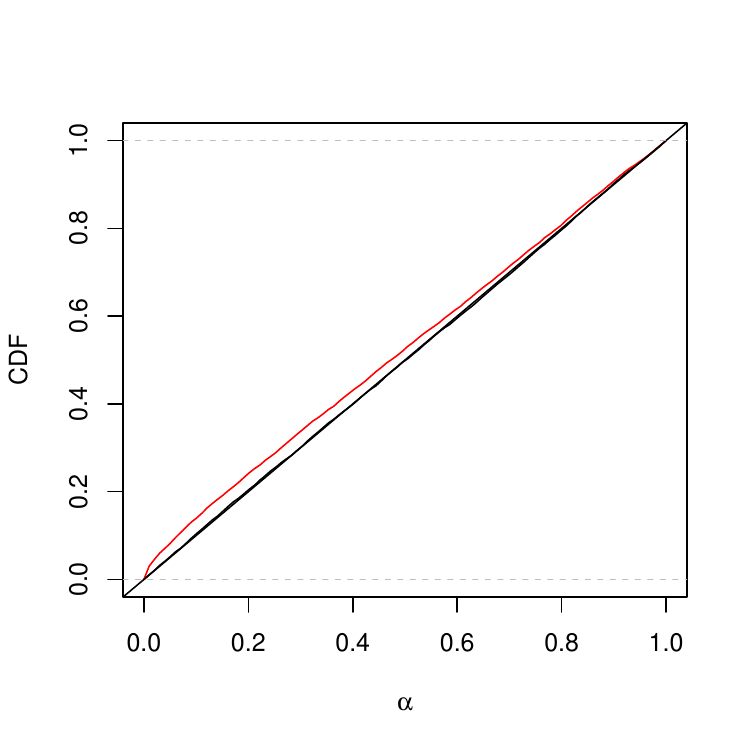}
\end{subfigure}
\caption{Panel~(a): Full data (black), individual (grey) and large-$n$ divide-and-conquer (red) possibility contours in the Exponential example. Panel~(b): Empirical distribution function of the full data (black) and large-$n$ divide-and-conquer (red) possibility contours in the Exponential example evaluated at $\Theta=0.5$ based on $10{,}000$ Monte Carlo samples.} \label{f:meta-IM}
\end{figure}

\begin{table}[t]
\centering
\caption{Empirical coverage probability (in \%) for the Exponential example simulation.} \label{t:gamma-CP}
\ra{0.8}
\begin{tabular}{rrrrrrrrrr}
Contour & \multicolumn{9}{c}{$100(1-\alpha)\%$} \\
& $10$ & $20$ & $30$ & $40$ & $50$ & $60$ & $70$ & $80$ & $90$\\
\midrule
$\pi_{s^n}^\infty$ & $9.180$ & $19.19$ & $28.78$ & $38.29$ & $47.40$ & $56.95$ & $66.19$ & $75.90$ & $86.08$ \\
$\pi_{s^n}^\forall$ & $9.790$ & $20.40$ & $30.47$ & $40.46$ & $50.17$ & $60.22$ & $69.91$ & $79.63$ & $89.66$
\end{tabular}
\end{table}

We have emphasized above that the large-$n$ divide-and-conquer IM is generally only valid asymptotically, but exact validity is not strictly impossible.  Of course, as our Gaussian-based motivation would suggest, the large-$n$ divide-and-conquer IM is exactly valid when the data are exactly Gaussian.  Moreover, it is easy to verify based on the formula in \eqref{eq:pi.naive} for $\pi_{s^n}^\infty$ that validity holds whenever the random variable $(\check\theta_{S^n} - \Theta)^\top J_{S^n} (\check\theta_{S^n} - \Theta)$ is stochastically no larger than $\chisq(p)$.  So, whenever the sampling distribution of $\check\theta_{S^n}$ is Gaussian or suitably ``sub-Gaussian,'' exact validity is expected.  This is an interesting observation, but ``sub-Gaussianity'' is far too restrictive for us to be satisfied, so we continue our quest.  The next section builds on this experience, offering a new divide-and-conquer IM that is exactly valid and asymptotically efficient. 

\section{Practical divide-and-conquer IMs} 
\label{s:proposal}

\subsection{Ranking via Gaussian working likelihoods}
\label{ss:middle}

Above we examined the large-$n$ divide-and-conquer IM that is based on simply mimicking the calculations that hold in the Gaussian case.  This alleviates all the inherent challenges associated with combining the block-specific information, since the optimal combination rule in the Gaussian case is well-established.  We also demonstrated that this large-$n$ divide-and-conquer IM asymptotically merges with the full-data IM.  Since the full data IM is both valid and efficient, the aforementioned merging implies that the large-$n$ divide-and-conquer IM is also {\em asymptotically} valid and efficient.  But a centerpiece of statistical inference is finite-sample validity---statisticians' Hippocratic Oath---so, to us, asymptotic validity is not enough.  Here we offer a middle ground strategy that takes some of what makes the large-$n$ divide-and-conquer IM good and leaves some of what makes it na\"{i}ve. This allows us to achieve the desired finite-sample validity with only a small increase in complexity compared to the large-$n$ divide-and-conquer IM.  

The key observation is that the large-$n$ divide-and-conquer IM uses a Gaussianity assumption twice: once in the ranking step, via the choice of relative likelihood, and again in the validification step.  That is, the proposal $\pi_{s^n}^\infty$ in \eqref{eq:pi.naive} can be expressed as 
\[ \pi_{s^n}^\infty(\theta) = \prob_\theta^\text{gauss}\{ R^\text{gauss}(S^n, \theta) \leq R^\text{gauss}(s^n, \theta) \}, \quad \theta \in \TT, \]
where the {\em Gaussian working relative likelihood} used for ranking is 
\begin{equation}
\label{e:gaussian-relative-likelihood}
R^\text{gauss}(s^n, \theta) = \exp\{-\tfrac12 (\check{\theta}_{s^n} - \theta)^\top J_{s^n} (\check{\theta}_{s^n} - \theta) \}, \quad \theta \in \TT, 
\end{equation}
and $\prob_\theta^\text{gauss}$ is the Gaussian model for $S^n$ used for validification, under which $(\check{\theta}_{S^n} - \theta)^\top J_{S^n} (\check{\theta}_{S^n} - \theta)$ has a $\chisq(p)$ distribution.  But we are under no obligation to use the same Gaussian assumption in both the ranking and validification steps. Its use in the ranking step is natural, as the derived relative likelihood incorporates all available information in $S^n$. However, applying it in the validification step forces us to replace the true distribution of the data with the Gaussian approximation. This is unnecessary because the validification based on the posited statistical model does not require the ranking to be derived from this model. In other words, the validification step can be applied to virtually any choice of ranking. 

Our proposal is then to use the Gaussian assumption only in the ranking step, through the choice of the Gaussian working relative likelihood in \eqref{e:gaussian-relative-likelihood}, and then carry out the validification step using the posited statistical model instead of the Gaussian approximation:
\begin{align}
\pi_{s^n}^\forall(\theta) = \prob_{\theta} \bigl\{ R^\text{gauss}(S^n, \theta) \leq R^\text{gauss}(s^n, \theta) \}, \quad \theta \in \TT,
\label{e:proposal-fused-IM}
\end{align}
where $\prob_{\theta}$ is the posited statistical model, which determines a sampling distribution of $S^n$ depending on the generic, hypothesized value $\theta$ of the unknown $\Theta$.  
This is clearly a bona fide possibility contour since it reaches a maximum value of 1 at $\check\theta_{s^n}$. It is also finite-sample valid, as demonstrated in Theorem~\ref{thm:middle_validity} below. We refer to the corresponding IM as the {\em valid divide-and-conquer IM}. 

To summarize, the large-$n$ divide-and-conquer IM is based on applying a Gaussian assumption in both the ranking and validification steps.  This results in a very simple solution, which happens to be asymptotically efficient, but is typically not valid in finite samples. Our proposed valid divide-and-conquer IM uses the Gaussian approximation only in the ranking step, which allows for ranking based solely on block summaries while preserving the exact validity of the full-data IM.  The trade-off is that there is no simple expression for $\pi_{s^n}^\forall$ like there was for $\pi_{s^n}^\infty$.  Despite this, the valid divide-and-conquer contour in \eqref{e:proposal-fused-IM} retains much of the simplicity and interpretability that made the large-$n$ divide-and-conquer approach appealing.
For example, the weighted average of maximum likelihood estimators in \eqref{e:meta-def} remains fully and most plausible: $\pi_{s^n}^\forall (\check\theta_{s^n}) =1$. 

\subsection{Validity and efficiency} \label{ss:middle-properties}

As the ``validification'' terminology suggests, when we carry out the validification step using the posited model, rather than some (Gaussian) approximation, the result is an IM that is exact valid, not just asymptotically so; see Theorem~\ref{thm:middle_validity}.  This implies that our proposal achieves what we referred to above as the statisticians' Hippocratic Oath.  

\begin{thm}
\label{thm:middle_validity}
The valid divide-and-conquer IM with contour as in \eqref{e:proposal-fused-IM} is valid in the sense of \eqref{eq:IMvalidity}. That is, $\prob_\Theta\{ \pi_{s^n}^\forall(\Theta) \leq \alpha \} \leq \alpha$, for all $\alpha \in [0,1]$. 
\end{thm}

An immediate consequence of the IM's validity is that the usual statistical procedures, i.e., hypothesis tests and confidence sets, control the frequentist error rates.  This is remarkable because, to our knowledge, no other divide-and-conquer procedure achieves this form of exact error rate control and at this level of generality.   

\begin{cor}[\citealt{martin2019false}] 
\label{cor:stat-inference}
The test that rejects the null hypothesis $H_0: \Theta \in A$ if and only if $\pi_{s^n}^\forall(A) \leq \alpha$ has size $\alpha$, and the set 
\begin{equation}
\label{eq:region}
C_{\alpha}^\forall(s^n) = \{ \theta \in \TT: \pi_{s^n}^\forall(\theta) > \alpha \}, \quad \alpha \in [0,1], 
\end{equation}
is a $100(1-\alpha)\%$ confidence region. That is,
\begin{align*}
\sup \limits_{\Theta \in A} \prob_{\Theta} ( \mbox{reject ``$H_0: \Theta \in A$''}) \leq \alpha \quad \text{and} \quad \inf \limits_{\Theta} \prob_{\Theta} \{ C_{\alpha}^\forall(S^n) \ni \Theta \} \geq 1-\alpha.
\end{align*}
\end{cor}

For a scalar $\Theta$ in a continuous $\prob_\Theta$, $C_{\alpha}^\forall(s^n)$ is obtained by reading off the $\theta$ values for which $\pi_{s^n}^\forall(\theta) = \alpha$, as illustrated in Appendix \ref{a:proposal}. Table~\ref{t:gamma-CP} reports the coverage probability of the corresponding confidence intervals, confirming their proper calibration.

While validity is crucial to the logic of statistical inference, this can be of very little practical value if the answers provided by valid methods are too conservative.  As a follow up to Theorem~\ref{thm:naive-efficiency}, next we establish that the finite-sample validity achieved by the proposed valid divide-and-conquer IM comes at no cost of statistical efficiency asymptotically.  That is, like the large-$n$ divide-and-conquer IM developed above, the proposed exactly-valid divide-and-conquer IM, with contour $\pi_{s^n}^\forall$, also merges asymptotically with the efficient, full-data IM, with contour $\pi_{Z^n}$.  

\begin{thm}\label{thm:middle_efficiency}
Under the same conditions as in Theorem~\ref{thm:naive-efficiency}, 
\begin{align*}
\sup_{u \in C} \bigl| \pi_{s^n}^\forall(\hat\theta_{Z^n} + J_{Z^n}^{-1/2} u) - \pi_{Z^n}(\hat\theta_{Z^n} + J_{Z^n}^{-1/2} u) \bigr| \to 0, \quad \text{in $\prob_{\Theta}$-probability, as $n \to \infty$},
\end{align*}
for any arbitrary compact set $C \subset \RR^p$. 
\end{thm}

\subsection{Computation and evaluation} \label{ss:middle-computation}

At face value, evaluating the contour $\pi_{s^n}^\forall$ in \eqref{e:proposal-fused-IM} seems to require potentially very expensive Monte Carlo computations.  That is, to evaluate $\pi_{s^n}^\forall(\theta)$, first choose a target number $M$ of Monte Carlo samples and then do the following for each $m=1,\ldots,M$: draw a copy of $Z^{(b)}$ from $\prob_\theta$ and evaluate the corresponding $S_b$ for each $b=1,\ldots,B$; aggregate $(S_1,\ldots,S_B)$ and call it $S_{m,\theta}^n$, with the subscript to indicate both the replicate and the particular $\prob_\theta$ from which it was drawn.  With the samples $\{S_{m,\theta}^n: m=1,\ldots,M\}$, it is straightforward to approximate the contour at $\theta$:
\[ \pi_{s^n}^\forall(\theta) \approx \frac1M \sum_{m=1}^M \mathbbm{1} \{ R^\text{gauss}(S_{m,\theta}^n, \theta) \leq R^\text{gauss}(s^n, \theta)\}. \]
Two things make this potentially expensive, which we will deal with separately. The first is that one apparently needs a different set of Monte Carlo samples corresponding to each $\theta$ being plugged into the contour and the computational cost can quickly escalate even for a moderate dimensional $\Theta$. The second is that obtaining $M$ Monte Carlo replicates of the block-specific summary statistics can be very expensive. We first propose a strategy to address the first issue and revisit the second one below.

The Gaussian working relative likelihood is at least an approximate pivot, so the sampling distribution of $R^\text{gauss}(S^n, \theta)$, as a function of $S^n$ under $\prob_{\theta}$, is nearly constant in $\theta$.  In fact, in both the Gaussian and Exponential models presented earlier, the Gaussian working relative likelihood is an exact pivot.  This property allows for an efficient approximation of $\pi_{s^n}^\forall(\cdot)$ across the entire parameter space using just a single Monte Carlo sample of $S^n$ corresponding to a convenient choice $\theta^\dagger$---say $\check{\theta}_{s^n}$---of the unknown parameter $\Theta$. 
That is, for a target Monte Carlo size $M$, let $S_{m,\theta^\dagger}^n$ denote a sample of $S^n$ from its sampling distribution under $\prob_{\theta^\dagger}$, for $m=1,\ldots,M$.  Then approximate the contour by 
\begin{equation}
\label{eq:contour.mc}
\pi_{s^n}^\forall(\theta) \approx \frac{1}{M} \sum_{m=1}^M \mathbbm{1} \{ R^\text{gauss}(S_{m,\theta^\dagger}^n, \theta^\dagger) \leq R^\text{gauss}(s^n, \theta)\}, \quad \theta \in \TT. 
\end{equation}

A bit more generally, if the Gaussian working relative likelihood is only an approximate pivot, then the sampling distribution of $S^n$ under $\prob_\theta$ depends weakly on $\theta$. This may lead to a not-so-accurate approximation in \eqref{eq:contour.mc} when the evaluation point $\theta$ is not too close to the anchor $\theta^\dagger$. If the simple strategy above only works well ``locally,'' then it could be made more flexible by spreading a few anchors $\theta_1^\dagger, \theta_2^\dagger, \ldots$ around the parameter space (e.g., using a space-filling design) and then, to approximate the contour $\pi_{s^n}^\forall(\theta)$ at a particular $\theta$, choose the anchor closest to $\theta$ and apply the corresponding Monte Carlo approximation as in \eqref{eq:contour.mc} above. Alternative computational approaches for cases where there may be concerns about how close the Gaussian working relative likelihood is to being an approximate pivot are described in Appendix \ref{a:importance-sampling}.

We now discuss how to draw samples of the summary statistics $S_m^n=(S_{m,1},\ldots,S_{m,K})$ for $m=1,\ldots,M$. While obtaining $\hat{\theta}_{z^{(b)}}$ is certainly easier than obtaining $\hat{\theta}_{Z^n}$, we now require $M$ such computations for the approximation in equation \eqref{eq:contour.mc}, which may become very expensive. Furthermore, obtaining even a single evaluation of the summary statistic $S_{m,b}$ may be computationally infeasible when the likelihood is not only expensive but intractable. The strategy we propose to overcome this issue is to learn the expected likelihood ``map'' from data $Z^{(b)} \sim \prob_{\theta}$ to $\theta$ using black-box emulators, such as normalizing flows, and simulations from the posited model. To learn this map, we first draw training values of $\Theta$ from some training distribution and sample data from the model given the training values. Using these sampled values, we then train a black-box algorithm that learns the map from data to parameter values. Once the algorithm is trained, the observations $z^{(b)}$ are fed into the trained emulator to obtain draws from the likelihood emulator. We compute $\hat{\theta}_{z^{(b)}}$ and $J_{z^{(b)}}$ as the mean and inverse variance, respectively, of these draws. The inference is said to be ``amortized'' \citep{sainsbury2024likelihood, zammit2024neural, hector2024whole} because, once the up-front training cost of the emulator has been paid, evaluations of $S_b$ and $S_{m,b}$ are virtually free.

The introduction of our emulator raises an important question: if we can obtain maximum likelihood estimates using amortized inference, why not train a black-box emulator to do the same on the whole data? This is not a viable alternative for two reasons. First, while it may be feasible to obtain the full data maximum likelihood estimator $\hat{\theta}_{Z^n}$ using such an amortized approach, the difficulty of computing the full-data relative likelihood, $R(Z^n, \theta)$ remains. While some density estimation techniques can be used to estimate $L_{z^n}(\theta)$ by smoothing out draws from the emulator, this smoothing will suffer from the curse of dimensionality for even moderately sized problems.  Second, even if this were feasible, it would become prohibitively expensive to perform this smoothing for $M$ Monte Carlo replicates to evaluate $\pi_{z^n}(\theta)$ in equation \eqref{e:oracle-IM}. This comparison emphasizes the importance of using summaries $S_1, \ldots, S_B$ and the Gaussian working relative likelihood $R^\text{gauss}(s^n, \theta)$ in defining $\pi_{s^n}^\forall(\theta)$.

\subsection{Profile likelihood for marginal contours} \label{ss:profile}

The divide-and-conquer IM is designed to provide reliable (and fully conditional) uncertainty quantification about $\Theta$ in light of the observed summary statistics.  But it is often the case that interest is in some feature of $\Theta$ rather than $\Theta$ in its entirety.  Of course, uncertainty quantification about $\Theta$ implies that about a feature of $\Theta$, but there is still the question of how to achieve the latter in the proposed IM framework.  The general IM marginalization strategy, which we follow here, is described in \cite{martinpp3}.  For simplicity, we will focus on the case of marginal inference on a component $\Theta_q$, taking values in $\TT_q$, of the full parameter $\Theta$ in $\TT$.  

Let $\check\theta_{s^n,q}$ denote the $q^\text{th}$ component of the estimator $\check\theta_{s^n}$, and let $J_{s^n,q}$ denote the $q^\text{th}$ entry on the diagonal of the observed information matrix $J_{s^n}$.  Also, for a generic $\theta$, write $(\theta_q, \theta_{-q})$ for the $q^\text{th}$ component and everything else.  Then it is easy to verify that 
\begin{align*}
R_q^\text{gauss}(s^n, \theta_q) & = \sup_{\theta_{-q}} R^\text{gauss}\bigl\{ s^n, (\theta_q, \theta_{-q}) \bigr\} \\
& = \exp\{-\tfrac12 (\check{\theta}_{s^n,q} - \theta_q)^\top J_{s^n,q} (\check{\theta}_{s^n,q}- \theta_q) \}, \quad \theta_q \in \TT_q. 
\end{align*}
This is just the profiled version of the Gaussian working relative likelihood.  Following \cite{martinpp3}, we propose to validify the profiled Gaussian working likelihood:
\begin{align}
\pi_{s^n}^{\forall,q}(\theta_q) &= \sup_{\theta_{-q}} \prob_{(\theta_q, \theta_{-q})} \{ R_q^\text{gauss}(S^n, \theta_q) \leq R_q^{\text{gauss}} (s^n, \theta_q) \}, \quad \theta_q \in \TT_q.
\label{e:proposal-fused-IM-q}
\end{align}
We have carried out the validification step several times already in this paper, so this is mostly familiar by now.  The one difference here is the outer supremum over the nuisance parameters $\theta_{-q}$.  This is needed because, of course, there is no ``$\prob_{\theta_q}$'' and, hence, no sampling distribution of $S^n$ that only depends on $\theta_q$.  That is, we need to fix the value of $\theta=(\theta_q, \theta_{-q})$ to determine the sampling distribution of $S^n$, but, since $\theta_{-q}$ is not an argument into the possibility contour, validity considerations force us to optimize over the unspecified $\theta_{-q}$; further details can be found in \cite{martinpp3}.  

As in Section \ref{ss:middle-computation}, the contour \eqref{e:proposal-fused-IM-q} can be approximated using Monte Carlo:
\[ \pi_{s^n}^{\forall,q}(\theta_q) \approx \sup_{\theta_{-q}} \frac1M \sum_{m=1}^M \mathbbm{1} \{ R_q^\text{gauss}(S_{m,\theta_q,\theta_{-q}}^n, \theta_q) \leq R_q^{\text{gauss}} (s^n, \theta_q) \}, \quad \theta_q \in \TT_q. \]
The supremum over nuisance parameters on the outside adds considerable complexity to this calculation.  However, as above, we still fully expect that there is an approximate pivot structure in the sampling distribution of the Gaussian working relative likelihood.  That is, we do not expect the distribution of $R^\text{gauss}(S_{m,\theta}^n, \theta)$ to vary much with $\theta$.  So, as before, we can fix an anchor $\theta^\dagger$ (or perhaps several anchors spread around the parameter space) and draw samples of $S^n$ from $\prob_{\theta^\dagger}$ exclusively.  To summarize, we propose to approximate the contour of our marginal, divide-and-conquer IM for $\Theta_q$ as 
\begin{equation}
\label{eq:contour.mc-q}
\pi_{s^n}^{\forall,q}(\theta_q) \approx \frac{1}{M} \sum_{m=1}^M \mathbbm{1} \{ R_q^\text{gauss}(S_{m,\theta^\dagger}^n, \theta_q^\dagger) \leq R_q^{\text{gauss}} (s^n, \theta_q) \},
\quad \theta_q \in \TT_q,
\end{equation}
where, as before, $\{S_{m,\theta^\dagger}^n: m=1,\ldots,M\}$ consists of samples of summary statistics $S^n$ drawn from $\prob_{\theta^\dagger}$ corresponding to the fixed anchor $\theta^\dagger$. The fact that this approximation strategy works well is supported empirically in our numerical examples of Section \ref{s:simulations}. Of course, if the approximate pivot structure does not hold, estimation of $\pi_{s^n}^{\forall,q}$ can still be carried out using Monte Carlo samples at all values $\theta$ of interest. We find, however, that this isn't necessary in our examples.

\section{Numerical examples} \label{s:simulations}

\subsection{L\'{e}vy's alpha-stable distributions} \label{ss:alpha-stable-sims}

Alpha-stable distributions \citep{levy1925calcul} are commonly used in finance, economics, and physics. A family of distributions is said to be alpha-stable if it is closed under convolution. These distributions are defined by their characteristic function,
\begin{align*}
\phi(t) &= \exp[it\mu - |c t |^{\alpha} \{ 1-i \beta \, \text{sgn}(t) \, \Phi\} ],
\end{align*}
where $i=\sqrt{-1}$ is the imaginary unit, $\text{sgn}(t)$ is the sign of $t$, and $\Phi=\tan(\pi \alpha/2)$ if $\alpha \neq 1$ and $\Phi = -2\log|t|/\pi$ if $\alpha=1$. The parameter $\alpha \in (0,2]$ is a stability parameter, $\mu\in \RR$ is a location parameter, $c\in (0,\infty)$ is a scale parameter and $\beta \in [-1,1]$ is a ``skew'' parameter. The density is recovered through the inverse Fourier transform, $(2\pi)^{-1} \int_{-\infty}^{\infty} \phi(t) e^{-ixt} \, dt$,
from which it is evident that alpha-stable densities do not in general have closed-forms. The Cauchy, L\'{e}vy and Gaussian distributions are prominent examples of closed form solutions for specific values of $\alpha, \beta, \mu$ and $c$.

We consider the setup with $n=200$, $B=4$, $n_b \equiv 50$. We take $\alpha=1.5$ to be fixed, and data points are generated from the alpha-stable distribution with $\mu=0$, $c=0.5$ and $\beta=0$ and using an algorithm proposed in \cite{chambers01061976} and summarized in Appendix \ref{a:alpha-stable}. Let $\Theta=(\mu, c, \beta)$ the true value of the location, scale and skew parameters. Appendix~\ref{a:alpha-stable-MLE} illustrates the lack of validity guarantee of large-sample inference based on the full maximum likelihood estimator, which motivates the evaluation of the computationally intractable valid contour described in Section~\ref{ss:oracle-IM} and subsequently the valid divide-and-conquer inference. 

We train an emulator to learn the map between data of size $n_b=50$ and parameter values of $\Theta$ using a training distribution of $\Theta$ that is continuous uniform on the interval $[-20,20]$, $[0,10]$ and $[-1,1]$ for $\mu$, $c$ and $\beta$ respectively. The emulator is based on two chained invertible neural networks, trained jointly in the DeepSets framework using the BayesFlow software \citep{bayesflow1, bayesflow2} so as to be invariant to permutations of the observations. The first network learns a ten-dimensional summary statistic from the $50$-dimensional data inputs, while the second network (consisting of six coupling layers) learns the parameters from the summary statistic. The emulator is trained using the online algorithm of \cite{bayesflow2} that samples from the model on-the-fly to improve generalization. Once the emulator is trained, we compute $\hat{\theta}_{z^{(b)}}$ and $J_{z^{(b)}}$ as the mean and inverse variance, respectively, of $1{,}000$ draws from the emulator. We compute $\pi_{s_n}^{\forall,q}$ using equation \eqref{eq:contour.mc-q} with $\theta^\dagger=\check{\theta}_{s_n}$ and $M=3{,}000$ Monte Carlo samples for $q\in \{1,2,3\}$. 

Figure \ref{f:alpha-stable-contour} plots the contours for the large-$n$ and valid divide-and-conquer IMs for one replicate, and Figure \ref{f:alpha-stable-ecdf} plots the empirical distribution function of the large-$n$ and valid divide-and-conquer possibility contours based on $1{,}000$ replicates of $\pi_{S_n}^{\forall,q}$. The large-$n$ IM is invalid as the empirical distribution function is far above the diagonal line.

\begin{figure}[t]
\centering
\includegraphics[width=0.75\textwidth]{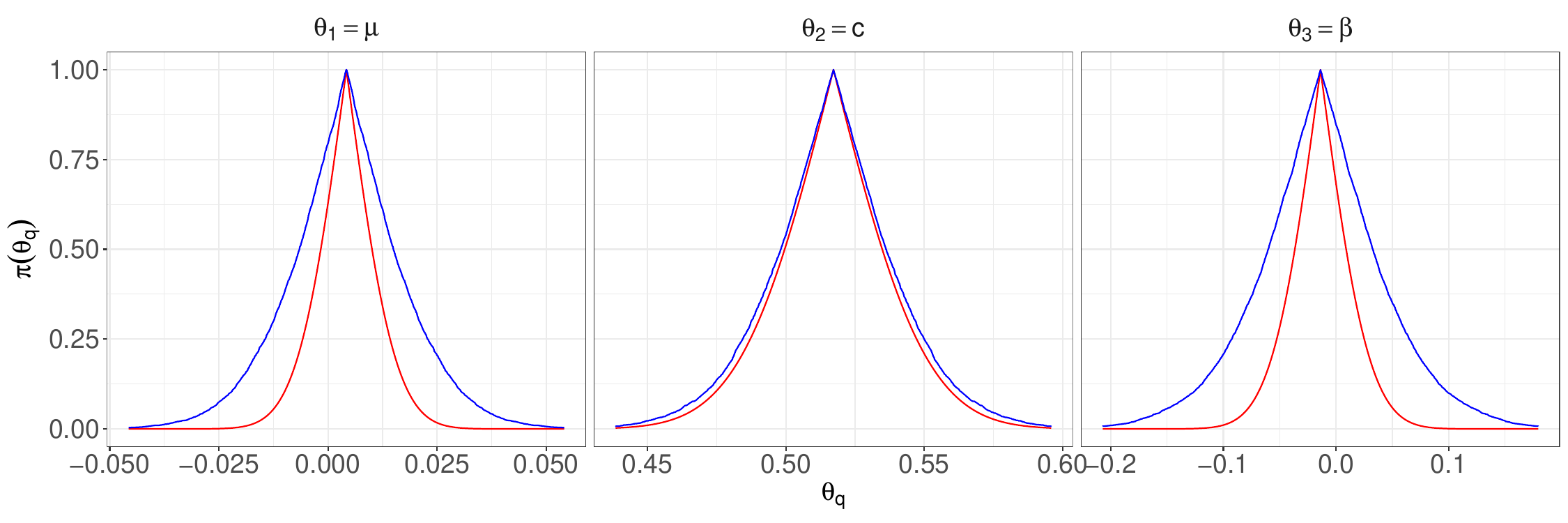}
\caption{Large-$n$ (red) and valid divide-and-conquer (blue) marginal possibility contours in the alpha-stable example.} \label{f:alpha-stable-contour}
\end{figure}

\begin{figure}[t]
\centering
\includegraphics[width=0.75\textwidth]{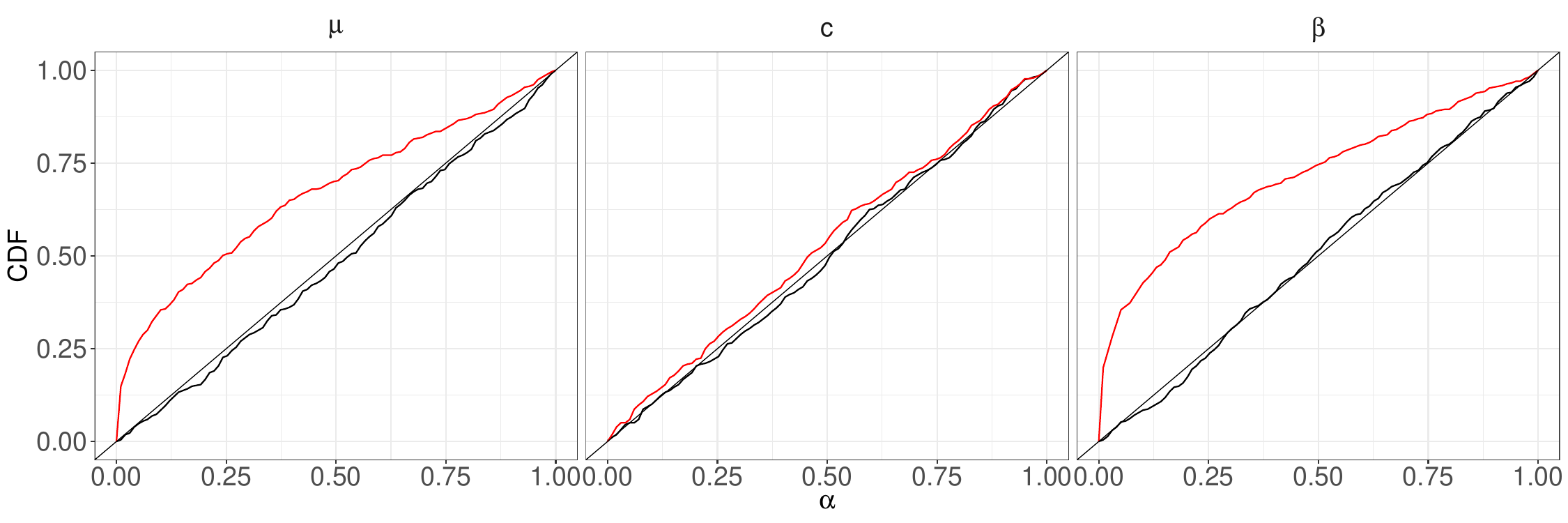}
\caption{Empirical CDF of the marginal valid divide-and-conquer (black) and large-$n$ (red) possibility contours in the alpha-stable example evaluated at $\Theta=(0,0.5,0)$ based on $1{,}000$ replicates.} \label{f:alpha-stable-ecdf}
\end{figure}

As an empirical check of Corollary \ref{cor:stat-inference}, we generate $1{,}000$ replicates of the large-$n$ and valid divide-and-conquer contours $\pi_{s^n}^\infty$ and $\pi_{s^n}^\forall$. For each replicate, we compute the $100(1-\alpha)\%$ marginal confidence intervals for $\Theta$ using $\{\theta_q \in \TT_q: \pi_{s^n}^{\infty,q}(\theta_q) > \alpha\}$ and $\{ \theta_q \in \TT_q: \pi_{s^n}^{\forall,q}(\theta_q)>\alpha\}$ at levels $\alpha \in \{0.1,0.2,\ldots,0.9\}$. The empirical coverage probability at level $100(1-\alpha)\%$, reported in Table \ref{t:alpha-stable-CP}, is the proportion of the $1{,}000$ computed intervals that contain the true value, $\Theta$. The empirical coverage probabilities track the nominal levels for the valid divide-and-conquer contour within margins of Monte Carlo standard error, but the large-$n$ contour substantially undercovers the true values of $\Theta$ at all nominal levels. Figure \ref{f:alpha-stable-contour} suggests, and Table \ref{t:alpha-stable-AL} confirms, that the large-$n$ confidence intervals are too narrow: the average length of the confidence intervals from the valid divide-and-conquer contour are larger than those from the large-$n$ contour. 

\begin{table}[t]
\caption{Simulation metrics for the alpha-stable simulations.} \label{t:alpha-stable}
\centering
\begin{subtable}[t]{\textwidth}
\caption{Empirical coverage probability (in \%).} \label{t:alpha-stable-CP}
\centering
\ra{0.8}
\begin{tabular}{rrrrrrrrrr}
Contour & \multicolumn{9}{c}{$100(1-\alpha)\%$} \\
& $10$ & $20$ & $30$ & $40$ & $50$ & $60$ & $70$ & $80$ & $90$\\
\midrule
$\pi_{s^n}^{\infty,1}$ & 4.90 & 11.0 & 17.2 & 22.0 & 27.7 & 33.4 & 39.9 & 48.8 & 59.3 \\
$\pi_{s^n}^{\infty,2}$ & 8.20 & 17.9 & 26.6 & 34.7 & 44.9 & 55.5 & 64.1 & 73.9 & 85.1 \\
$\pi_{s^n}^{\infty,3}$ & 7.00 & 13.7 & 19.4 & 24.9 & 31.0 & 36.7 & 45.4 & 53.9 & 64.2 \\
\midrule
$\pi_{s^n}^{\forall,1}$ & 11.7 & 22.4 & 32.6 & 42.3 & 51.7 & 61.6 & 71.6 & 81.2 & 90.8 \\
$\pi_{s^n}^{\forall,2}$ & 8.90 & 20.6 & 28.7 & 37.4 & 49.0 & 59.6 & 69.3 & 78.2 & 88.7 \\
$\pi_{s^n}^{\forall,3}$ & 13.3 & 23.8 & 33.0 & 43.8 & 52.5 & 61.9 & 71.4 & 81.4 & 91.1 \\
\end{tabular}
\end{subtable}\\
\vspace*{1em}
\begin{subtable}[t]{\textwidth}
\caption{Average confidence region length $\times 100$.} \label{t:alpha-stable-AL}
\centering
\ra{0.8}
\begin{tabular}{rrrrrrrrrr}
Contour & \multicolumn{9}{c}{$100(1-\alpha)\%$} \\
& $10$ & $20$ & $30$ & $40$ & $50$ & $60$ & $70$ & $80$ & $90$\\
\midrule
$\pi_{s^n}^{\infty,1}$ & 0.847 & 1.71 & 2.6 & 3.54 & 4.56 & 5.69 & 7.01 & 8.67 & 11.1 \\
$\pi_{s^n}^{\infty,2}$ & 0.654 & 1.32 & 2.01 & 2.73 & 3.52 & 4.39 & 5.4 & 6.68 & 8.58 \\
$\pi_{s^n}^{\infty,3}$ & 0.210 & 0.423 & 0.644 & 0.877 & 1.13 & 1.41 & 1.73 & 2.15 & 2.75 \\
\midrule
$\pi_{s^n}^{\forall,1}$ & 1.81 & 3.66 & 5.52 & 7.40 & 9.39 & 11.7 & 14.4 & 17.9 & 23.0 \\
$\pi_{s^n}^{\forall,2}$ & 0.723 & 1.47 & 2.23 & 2.99 & 3.84 & 4.85 & 5.96 & 7.33 & 9.34 \\
$\pi_{s^n}^{\forall,3}$ & 0.402 & 0.805 & 1.22 & 1.65 & 2.11 & 2.63 & 3.23 & 4.01 & 5.16 \\
\end{tabular}
\end{subtable}
\end{table}

\subsection{\textit{g}-and-\textit{k} distributions} \label{ss:g-and-k-sims}

We return to the $g$-and-$k$ distribution introduced in Section \ref{s:intro}. Let $\Theta=(\mu,\sigma,g,k)$ the true value of the location, scale, skew and kurtosis parameters. As is common in the existing literature \citep[see, e.g.][]{rayner2002numerical, drovandi2011likelihood}, we set $c=0.8$. We consider the setup with $n=200$, $B=4$ and $n_b \equiv 50$. Outcomes $y_j$ are generated from the $g$-and-$k$ distribution with $\mu=3$, $\sigma=1$, $g=2$ and $k=0.5$ using the R package \verb|gk| \citep{gk-package}. Appendix \ref{a:g-and-k-MLE} shows that, in this case, large-sample inference based on the full maximum likelihood estimator appears valid, although there is no guarantee. To guarantee validity, computationally intractable validification of the likelihood ratio would be needed. We train an emulator to learn the map between data of sample size $n_b=50$ and parameter values of $\Theta$ using a training distribution of $\Theta$ that is continuous uniform on the interval $[-20,20]$, $[-20,20]$, $[-5,5]$ and $[-1/2,5]$ for $\mu, \sigma, g$ and $k$ respectively. The emulator and computation of $\hat{\theta}_{z^{(b)}}$ and $J_{z^{(b)}}$ are as described in Section \ref{ss:g-and-k-sims}. We compute $\pi_{s_n}^{\forall,q}$ using equation \eqref{eq:contour.mc-q} with $\theta^\dagger=\check{\theta}_{s_n}$ and $M=3{,}000$ Monte Carlo samples for $q\in \{1,2,3,4\}$. 

Figure \ref{f:g-and-k-contour} plots the contours for the large-$n$ and valid divide-and-conquer IMs for one replicate, and Figure \ref{f:g-and-k-ecdf} plots the empirical distribution function of the large-$n$ and valid divide-and-conquer possibility contours based on $1{,}000$ replicates of $\pi_{S_n}^{\forall,q}$. The large-$n$ IM is invalid as the empirical distribution function is far above the diagonal line.

\begin{figure}[t]
\centering
\includegraphics[width=\textwidth]{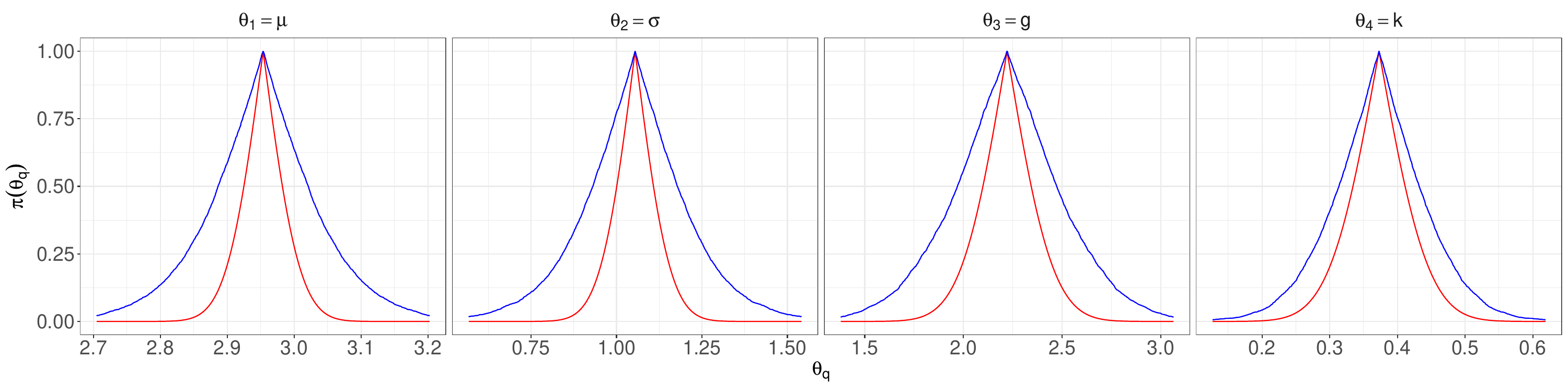}
\caption{Large-$n$ (red) and valid divide-and-conquer (blue) marginal possibility contours in the $g$-and-$k$ example.} \label{f:g-and-k-contour}
\end{figure}
\begin{figure}[t]
\centering
\includegraphics[width=\textwidth]{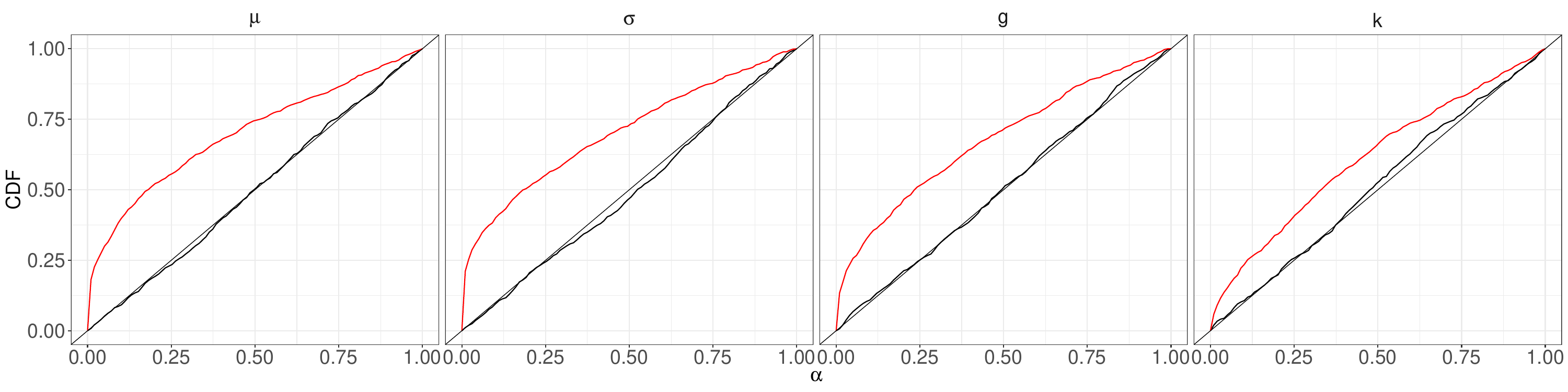}
\caption{Empirical CDF of the marginal valid divide-and-conquer (black) and large-$n$ (red) possibility contours in the $g$-and-$k$ example evaluated at $\Theta=(3,1,2,0.5)$ based on $1{,}000$ replicates.} \label{f:g-and-k-ecdf}
\end{figure}

As an empirical check of Corollary \ref{cor:stat-inference}, we generate $1{,}000$ replicates of the large-$n$ and valid divide-and-conquer contours $\pi_{s^n}^\infty$ and $\pi_{s^n}^\forall$. The marginal empirical coverage probability for $\Theta$ at level $100(1-\alpha)\%$, $\alpha \in \{0.1,0.2,\ldots,0.9\}$, reported in Table \ref{t:g-and-k-CP}, tracks the nominal level for the valid divide-and-conquer contour within margins of Monte Carlo standard error, but the large-$n$ contour substantially undercovers the true values of $\Theta$ at all nominal levels. Figure \ref{f:g-and-k-contour} suggests, and Table \ref{t:g-and-k-AL} confirms, that the large-$n$ confidence intervals are again too narrow. This highlights the importance of using our valid divide-and-conquer contour over the large-$n$ contour: if deployed in a real-world scenario where $\Theta$ is unknown, using the $90\%$ confidence interval based on the large-$n$ contour would unknowingly increase the type-I error rate by up to threefolds! In other words, there would be up to three times as many false discoveries.

\begin{table}[t]
\caption{Simulation metrics for the $g$-and-$k$ simulations.} \label{t:g-and-k}
\centering
\begin{subtable}[t]{\textwidth}
\caption{Empirical coverage probability (in \%).} \label{t:g-and-k-CP}
\centering
\ra{0.8}
\begin{tabular}{rrrrrrrrrr}
Contour & \multicolumn{9}{c}{$100(1-\alpha)\%$} \\
& $10$ & $20$ & $30$ & $40$ & $50$ & $60$ & $70$ & $80$ & $90$\\
\midrule
$\pi_{s^n}^{\infty,1}$ & 4.90 & 10.2 & 16.0 & 20.2 & 25.4 & 32.4 & 39.8 & 47.8 & 60.2 \\
$\pi_{s^n}^{\infty,2}$ & 4.80 & 9.20 & 14.3 & 19.6 & 27.0 & 33.3 & 41.1 & 48.9 & 60.2 \\
$\pi_{s^n}^{\infty,3}$ & 4.50 & 9.40 & 14.5 & 22.9 & 28.8 & 35.7 & 44.8 & 54.0 & 65.9 \\
$\pi_{s^n}^{\infty,4}$ & 7.00 & 13.7 & 19.8 & 26.4 & 33.5 & 43.0 & 53.4 & 65.8 & 76.5 \\
\midrule
$\pi_{s^n}^{\forall,1}$ & 9.80 & 19.8 & 30.4 & 41.1 & 52.5 & 63.6 & 73.9 & 82.7 & 92.0 \\
$\pi_{s^n}^{\forall,2}$ & 10.0 & 22.1 & 33.9 & 44.8 & 55.8 & 65.6 & 74.1 & 85.0 & 93.9 \\
$\pi_{s^n}^{\forall,3}$ & 8.20 & 19.3 & 29.8 & 41.3 & 52.6 & 62.4 & 72.0 & 81.3 & 90.4 \\
$\pi_{s^n}^{\forall,4}$ & 10.3 & 19.2 & 27.5 & 37.8 & 48.8 & 60.8 & 71.0 & 80.6 & 90.6 \\
\end{tabular}
\end{subtable}\\
\vspace*{1em}
\begin{subtable}[t]{\textwidth}
\caption{Average confidence region length $\times 100$.} \label{t:g-and-k-AL}
\centering
\ra{0.8}
\begin{tabular}{rrrrrrrrrr}
Contour & \multicolumn{9}{c}{$100(1-\alpha)\%$} \\
& $10$ & $20$ & $30$ & $40$ & $50$ & $60$ & $70$ & $80$ & $90$\\
\midrule
$\pi_{s^n}^{\infty,1}$ & 1.04 & 2.11 & 3.21 & 4.37 & 5.63 & 7.02 & 8.65 & 10.7 & 13.7 \\
$\pi_{s^n}^{\infty,2}$ & 2.03 & 4.11 & 6.25 & 8.51 & 10.9 & 13.7 & 16.8 & 20.8 & 26.7 \\
$\pi_{s^n}^{\infty,3}$ & 5.54 & 11.2 & 17.0 & 23.2 & 29.8 & 37.2 & 45.9 & 56.7 & 72.8 \\
$\pi_{s^n}^{\infty,4}$ & 1.77 & 3.58 & 5.45 & 7.41 & 9.54 & 11.9 & 14.7 & 18.1 & 23.3 \\
\midrule
$\pi_{s^n}^{\forall,1}$ & 2.06 & 4.21 & 6.52 & 8.99 & 11.6 & 14.5 & 18.0 & 22.6 & 29.7 \\
$\pi_{s^n}^{\forall,2}$ & 4.57 & 9.29 & 14.1 & 19.1 & 24.5 & 30.5 & 37.7 & 46.6 & 60.5 \\
$\pi_{s^n}^{\forall,3}$ & 10.2 & 20.8 & 31.6 & 42.7 & 54.0 & 66.4 & 80.5 & 98.4 & 125 \\
$\pi_{s^n}^{\forall,4}$ & 2.55 & 5.13 & 7.79 & 10.5 & 13.5 & 16.9 & 20.7 & 25.5 & 33.1 \\
\end{tabular}
\end{subtable}
\end{table}

\subsection{PM 2.5 data analysis} \label{ss:pm25}

PM\textsubscript{2.5} refers to particles with a diameter of 2.5 micrometers or less that, due to their small size, can be absorbed into the bloodstream and cause serious health problems. Wildfires are a significant source of PM\textsubscript{2.5} and their prevalence is expected to continue increasing with climate change \citep{chen2021mortality}. Exposure to high concentrations of PM\textsubscript{2.5} from wildfire smoke was found to have an association with birthweight \citep{birtill2024effects} and emergency department admissions in Australia \citep{ranse2022impact}. The Australian bushfire season from August to December of 2023 made international headlines \citep{ABCnews} and burned approximately $84$ million hectares \citep{PreventionWeb}, including several fires in Queensland. Due to the dangerous consequences of exposure to high levels of PM\textsubscript{2.5} on health, we model the distribution of daily medians of PM\textsubscript{2.5} in the city of Maryborough in Queensland, Australia, as a function of season to better understand windows of exposure for its residents. The data consists of daily medians of the hourly average of PM\textsubscript{2.5} (in micrograms per cubic metre) measured at their Maryborough site from January 1 to December 31, 2023. The data are publicly available under the Creative Commons Attribution 4.0 license and available for download on the Queensland government open data portal \citep{PM25-data}.

A histogram and scatter plot of the $n=365$ daily medians are plotted in Figure \ref{f:PM25-hist}. Let $(y_j)_{j=1}^{365}$ be the daily medians of PM\textsubscript{2.5} over the course of the year, with $y_j$ assumed to follow a $g$-and-$k$ distribution with location $\mu$, scale $\sigma_j$, skew $g$ and kurtosis $k$ ($c=0.8$). To fit the time trend, we model the scale parameter using a linear expansion of B-splines of degree five with knots at $365/3, 365/2$ and $2\times 365/3$:
\begin{align*}
\log \sigma_j &= \sum_{r=-5}^2 \beta_r \psi_{r,5}(j), \quad j=1, \ldots, 365,
\end{align*}
with $\psi_{r,5}(j)$ the B-splines and $\beta_{-5},\ldots, \beta_2$ the unknown B-spline coefficients, with $\psi_{0,5}(j)=1$ and $\beta_0$ acting as an intercept. A plot of the B-splines is provided in Appendix~\ref{a:data}. 

\begin{figure}[t]
\centering
\includegraphics[width=\textwidth]{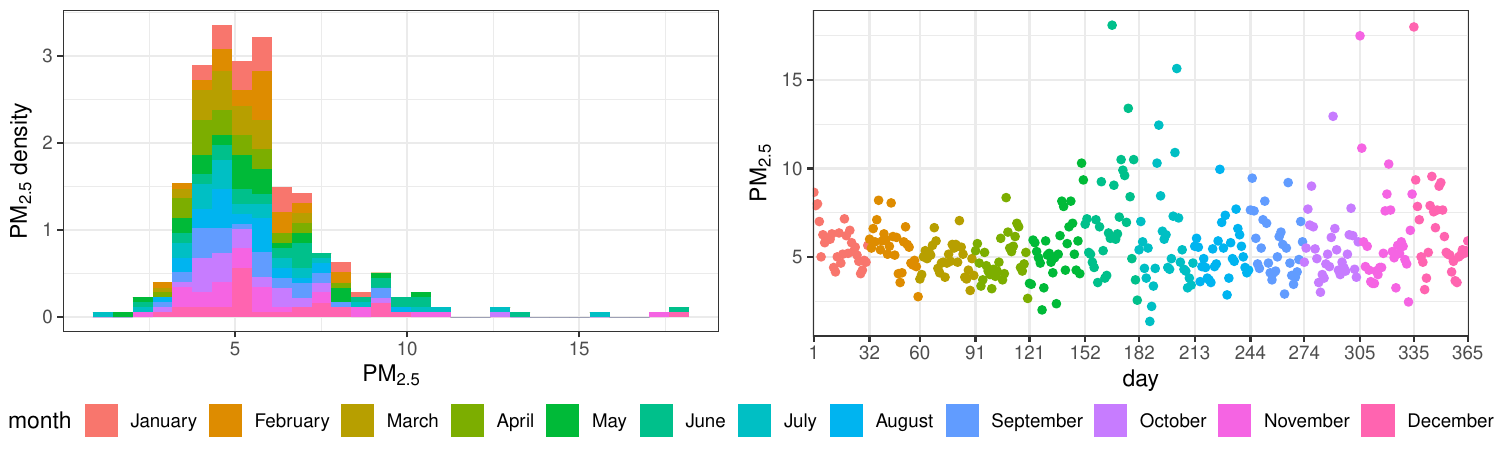}
\caption{Histogram and scatter plot of PM\textsubscript{2.5} daily median in the city of Maryborough.} \label{f:PM25-hist}
\end{figure}

The first uncertain parameter of interest is $\theta^\star=(\mu,\beta_{-5},\ldots,\beta_2,g,k) \in \RR^{11}$. We randomly shuffle the data $z^n = \{ \psi_{-5,5}(j), \ldots, \psi_{2,5}(j), y_j \}_{j=1}^{365}$ and partition the shuffled data into $B=4$ blocks $z^{(b)}$ of sizes $n_1=n_2=n_3=91$, $n_4=92$. For each block, we train an emulator to learn the map between data $z^{(b)}$ of size $n_b$ and values of $\theta^\star$ using a training distribution of $\theta^\star$ that is uniform on the interval $[-20,20]$ for $\mu$, $[-2,2]$ for $\beta_{-5},\ldots,\beta_2$, $[-5,5]$ for $g$ and $[-1/2,5]$ for $k$ (the training distribution gives a range of $\sigma_j$ values from $0.0183$ to $54.6$). The emulator is as described in Section \ref{ss:alpha-stable-sims} but with a $25$-dimensional summary statistic. Once the emulator is trained, we compute $\hat{\theta}_{z^{(b)}}^\star$ and $J_{z^{(b)}}^\star$ as the mean and inverse variance, respectively, of $1{,}000$ draws from the emulator, and $\check{\theta}_{s_n}^\star$ as in equation \eqref{e:meta-def}. From here, we can also compute a large-$n$ estimate of the second and primary uncertain parameter of interest, $\theta=(\mu,\sigma_1, \ldots, \sigma_{365},g,k)$, using 
\begin{align*}
\check{\theta}_{s^n} &= \Bigl( \check{\theta}_{s^n,1}, \bigl[ \exp \{ \psi_{-5,5}(j)\check{\theta}_{s^n,3} + \ldots \psi_{2,5}(j) \check{\theta}_{s^n,9} \} \bigr]_{j=1, \ldots, 365}, \check{\theta}_{s^n,10},\check{\theta}_{s^n,11} \Bigr) \in \RR^{368}.
\end{align*}
An estimate of its asymptotic inverse variance is $J_{s^n} = d^\top J_{s^n}^\star d$ with $J_{s^n}^\star=\sum_{b=1}^B J_{z^{(b)}}^\star$ and
\begin{align*}
d_{rj} &= \{0, \check{\theta}_{s^n,1+j} \psi_{r,5}(j), 0, 0\} \in \mathbb{R}^{11}, \quad j=1, \ldots, 365, \quad r=-5, \ldots, 2,\\
d &= \bigl\{1_1, ( d_{-5j}, \ldots, d_{2j} )_{j=1}^{365}, 1_{10}, 1_{11} \bigr\} \in \mathbb{R}^{11\times 368},
\end{align*}
and $1_q \in \mathbb{R}^{11}$ is the $q^\text{th}$ standard basis vector. We compute $\pi_{s_n}^{\forall,q}$ using equation \eqref{eq:contour.mc-q} with $\theta^\dagger=\check{\theta}_{s_n}$ and $M=5{,}000$ Monte Carlo samples for $q\in \{1,\ldots,368\}$. 

One of the unique features of the proposed valid divide-and-conquer IM framework is that finite-sample valid inference can be carried out on the $365$ daily scale parameters. Figure \ref{f:PM25-scale-fits} plots the large-$n$ estimates of daily scales, $\check{\theta}_{s^n,2}, \ldots, \check{\theta}_{s^n,366}$, with $90\%$ marginal confidence intervals constructed using $\{\theta_q \in \TT_q: \pi_{s_n}^{\forall,q}(\theta_q)> 0.1 \}$. The point estimates appear to mimic the pattern of daily median PM\textsubscript{2.5} observations in the scatterplot of Figure \ref{f:PM25-hist}. A plot of the observed versus fitted quantiles (obtained via probability integral transform using the fitted distribution function of the $g$-and-$k$ distribution) in Figure \ref{f:PM25-PITs} suggests that our model fits well. Figure \ref{f:PM-other-IMs} plots the large-$n$ and valid divide-and-conquer possibility contours for the location, skew and kurtosis parameters. The $90\%$ confidence intervals for the location, skew and kurtosis are $(5.198, 5.218)$, $(0.6534, 0.6811)$ and $(0.1762, 0.1955)$, indicating a positive skew and heavy right tail consistent with the histogram in Figure \ref{f:PM25-hist}. 

\begin{figure}[t]
\centering
\begin{subfigure}[t]{0.68\textwidth}
\includegraphics[width=\textwidth]{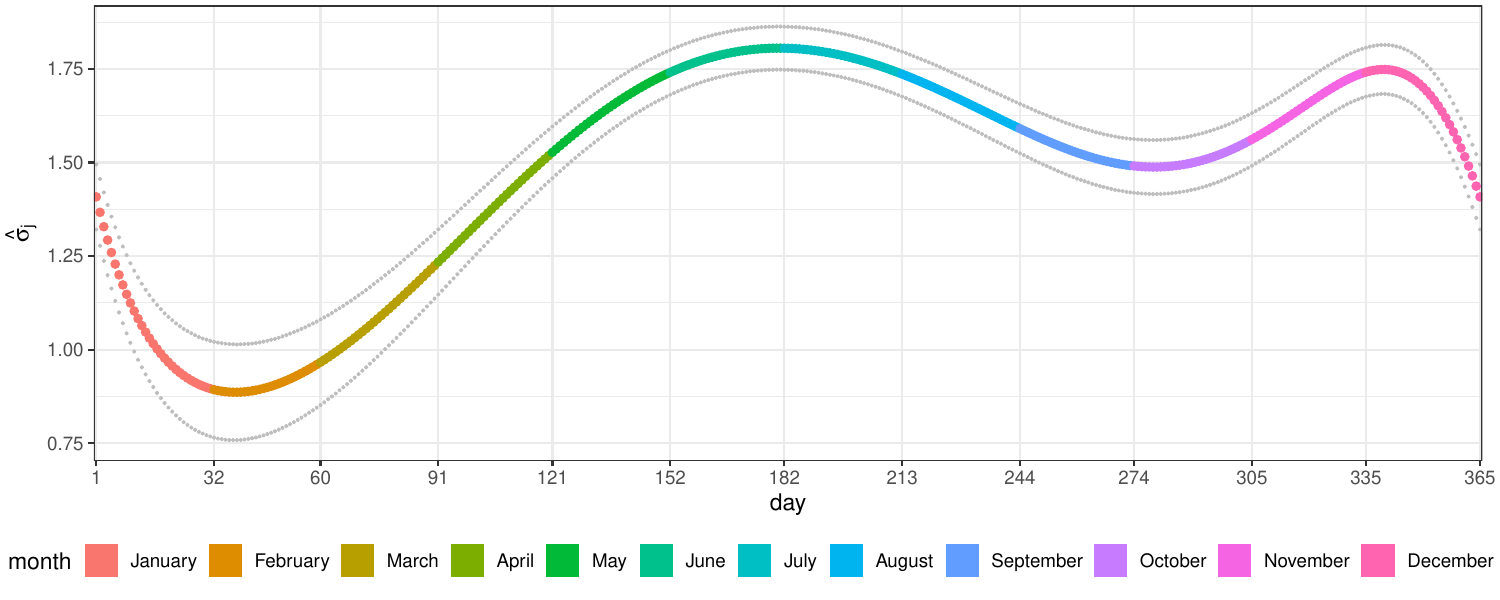}
\caption{} \label{f:PM25-scale-fits}
\end{subfigure}
\begin{subfigure}[t]{0.28\textwidth}
\includegraphics[width=\textwidth]{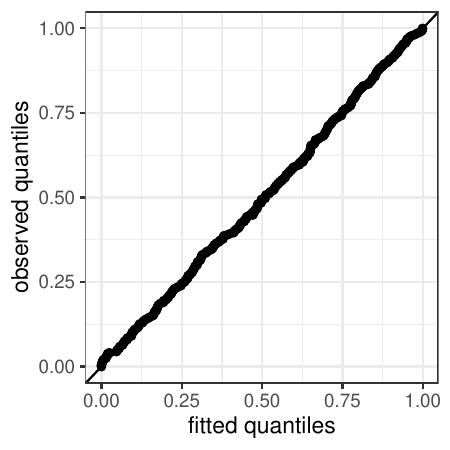}
\caption{} \label{f:PM25-PITs}
\end{subfigure}
\caption{(a) Scatter plot of the large-$n$ estimate of daily fitted scale parameters with marginal $90\%$ confidence interval in grey. (b) Observed versus fitted quantiles for the PM\textsubscript{2.5} analysis.}
\end{figure}

\begin{figure}[t]
\centering
\includegraphics[width=0.75\textwidth]{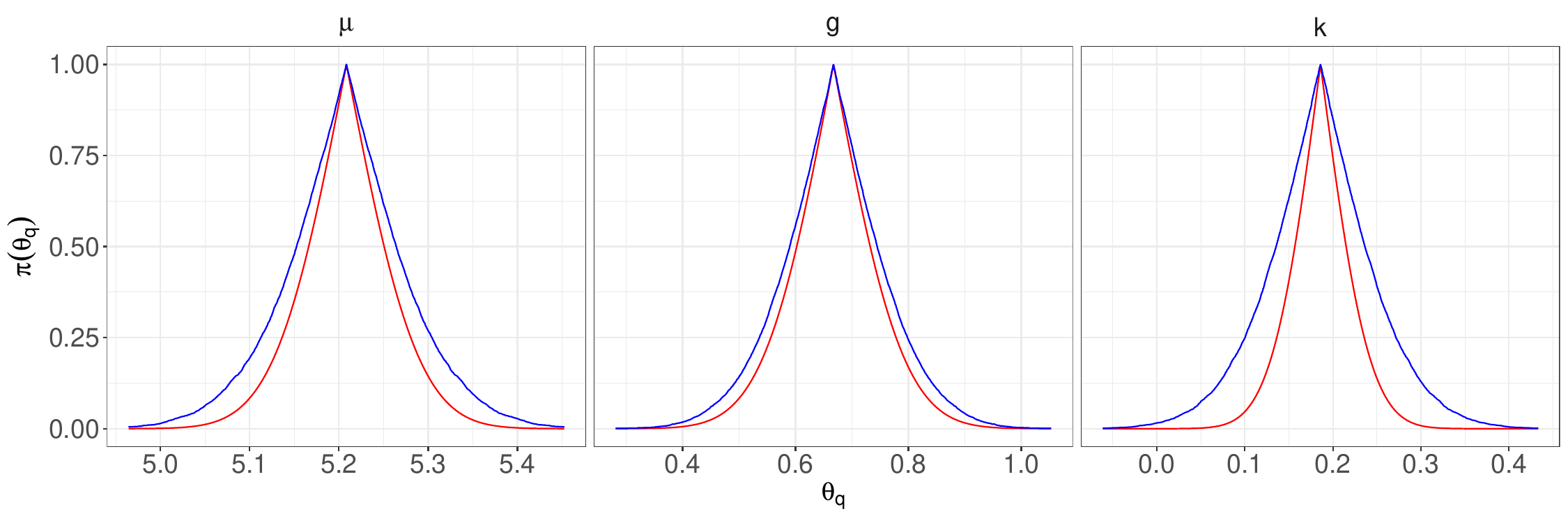}
\caption{Large-$n$ (red) and valid divide-and-conquer (blue) marginal possibility contours of $\mu$, $g$ and $k$ for the PM\textsubscript{2.5} analysis.} \label{f:PM-other-IMs}
\end{figure}

Figure \ref{f:PM25-scale-IMs} plots the large-$n$ and valid divide-and-conquer possibility contours for the daily scale parameters $\sigma_1, \ldots, \sigma_{365}$. As suggested in Figure \ref{f:PM25-scale-fits}, there appears to be greater variability in the scales across days in December and January (summer), and comparatively less in June and July (winter). As in Section \ref{ss:g-and-k-sims}, the large-$n$ possibility contour is improperly calibrated for inference. In contrast, the valid divide-and-conquer contour is wider because it appropriately accounts for the amount of information in the observed sample. The uncertainty in the model is largest in January, February and March, as evidenced by the wider $90\%$ confidence intervals in Figure \ref{f:PM25-scale-fits}. The largest values of the scales correspond to winter months (June and July) and early summer (November and December), meaning that Maryborough residents are at highest risk of exposure to PM\textsubscript{2.5} during these months. 

\begin{figure}[t]
\centering
\includegraphics[width=0.95\textwidth]{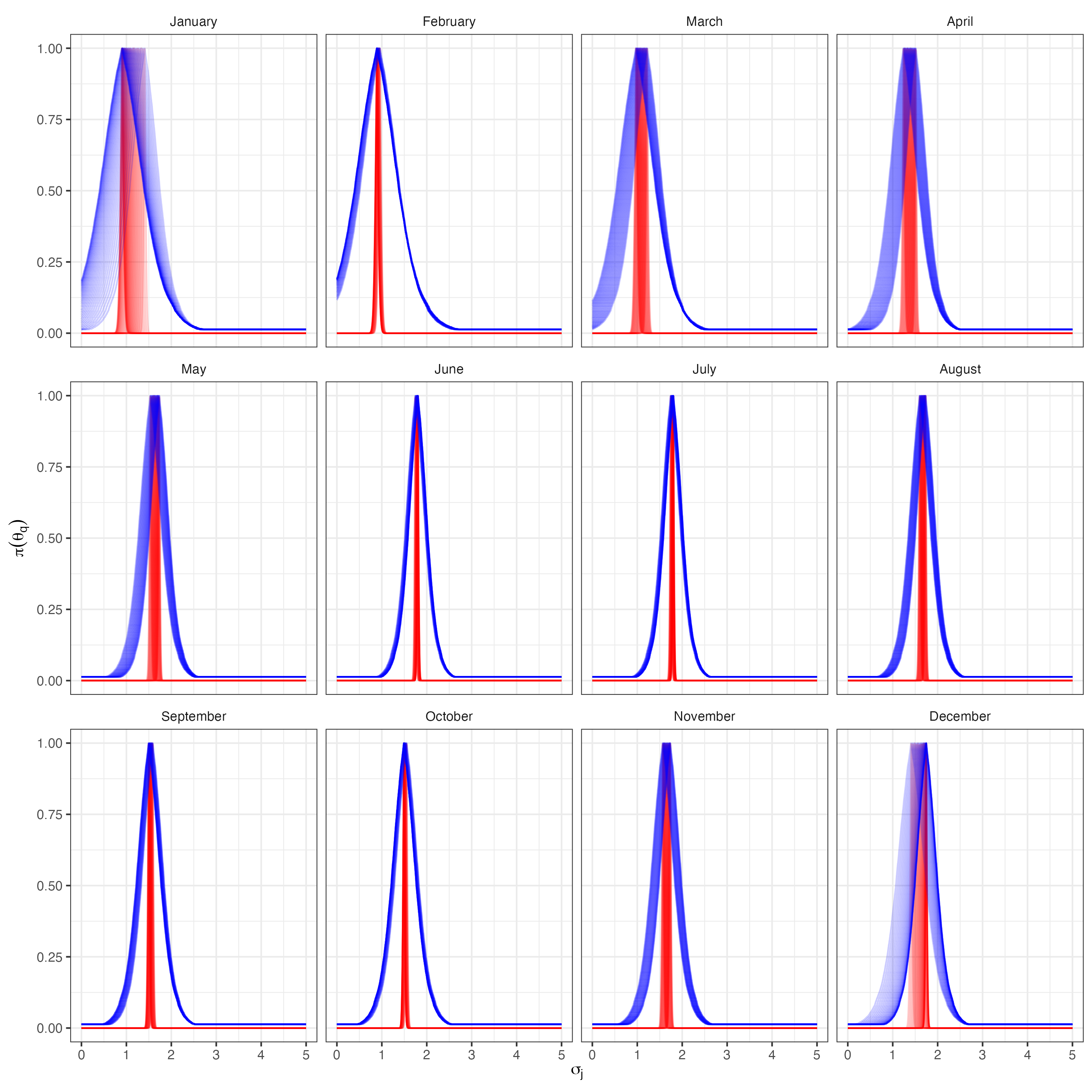}
\caption{Large-$n$ (red) and valid divide-and-conquer (blue) marginal possibility contours of the daily scale parameter in each month for the PM\textsubscript{2.5} analysis.} \label{f:PM25-scale-IMs}
\end{figure}

\section{Conclusion}

The large-$n$ divide-and-conquer possibility contour presented in Section~\ref{s:asymptotic-rule} leverages (approximate) Gaussianity in both the ranking and validification steps of the IM construction. In Theorem~\ref{thm:naive-efficiency}, we show that this construction is asymptotically valid and efficient (i.e., agrees with the optimal full-data, likelihood-based IM), but this is insufficient for our purposes since we are motivated by settings with moderate $n$. The valid divide-and-conquer IM employs the Gaussian relative likelihood in the ranking step, with the key distinction that the validification step is carried out using the true sampling distribution of the summaries. As the name suggests, this ensures the validity of the valid divide-and-conquer IM even in finite-sample settings. We demonstrate in Theorem~\ref{thm:middle_efficiency} that the valid divide-and-conquer IM is also asymptotically efficient, meaning that we gain important reliability guarantees at no asymptotic loss relative to both the full data and large-$n$ divide-and-conquer IMs. 

Our focus was on the construction of a valid and efficient divide-and-conquer IM, described by its possibility contour, which we then used in ways that have a subtle and perhaps unexpected Bayesian flavor.  One proposed use of the IM contour is to \emph{visually} show what the data have to say about (the relevant features of) the uncertain $\Theta$, e.g., in Figure~\ref{f:alpha-stable-contour}, offering a frequentists' visual counterpart to the Bayesians' posterior density.  Another proposed use of the IM contour is procedural, for reading off confidence sets as in Equation~\eqref{eq:region}, akin to how Bayesians read off highest posterior density credible sets, the key difference being that a valid IM's level sets are automatically calibrated to be confidence sets.  Beyond these Bayesian--frequentist connections, we also mentioned several times that the IM output is not just a tool for reading off confidence sets but, rather, can be used for fully conditional, data-dependent, probability-like uncertainty quantification about $\Theta$, comparable to a Bayesian posterior distribution.  Indeed, one can do formal inference by computing upper/lower probabilities associated with relevant hypotheses about $\Theta$, and, more generally, upper/lower expectations of functions of $\Theta$ for, say, formal decision-theoretic evaluation of relevant actions.  Importantly, validity of the IM implies that all this Bayesian-like uncertainty quantification, not just confidence sets, is reliable or calibrated in a frequentist sense.

While our proposal assumes the analyst seeks the most statistically efficient inference and hence wants to combine maximum likelihood estimators, the Gaussian relative likelihood $R^\text{gauss}$ in equation \eqref{e:gaussian-relative-likelihood} acts as a convenient combination rule that can also serve to combine other estimators. For example, one could replace the block-specific maximum likelihood estimator $\hat{\theta}_{z^{(b)}}$ with, say, a method of moments estimator, and the observed Fisher information $J_{s^n}$ with any symmetric positive-definite matrix in the formulation of $R^\text{gauss}$. This could be particularly useful in settings where a likelihood does not exist and training the emulator is computationally expensive. One drawback is that the resulting valid divide-and-conquer IM loses its asymptotic efficiency (Theorem \ref{thm:middle_efficiency}). Although well motivated by the connection to the large-$n$ divide-and-conquer estimator, another interesting direction for future work is to explore alternatives to the Gaussian relative likelihood for combining block-specific estimators.

One limitation, which has perhaps not received sufficient attention in our paper, is the difficulty of computing joint possibility contours when the dimension of $\Theta$ is large. This is most notable in the analysis of median PM\textsubscript{2.5} in Section \ref{ss:pm25}, for which we compute marginal daily confidence regions for the scale parameters $\sigma_j$ instead of a joint confidence set for all scale parameters $\sigma_1, \ldots, \sigma_{365}$. The primary challenge remains computation due to the dimension of the grid of $\theta$ values at which to evaluate the possibility contour $\pi_{s^n}^\forall$ using equation \eqref{eq:contour.mc}. Recent work on probabilistic approximations to possibility contours \citep{martin2025new} may offer some new directions in high dimensions. 

\section*{Acknowledgments}

This work is partially supported by the U.~S.~National Science Foundation, grants DMS--2337943 (ECH), SES–2051225 (RM), and DMS–2412628 (RM).

\appendix

\section{Details from Section~\ref{ss:thule}}
\label{p:gamma-IM-proof}

Recall that, in Example~\ref{ss:expo}, the primary goal was to evaluate the block-specific possibilistic IM contours, i.e., $\pi_{z^{(b)}}(\theta)$ in the Exponential model; the secondary goal is to evaluate the full data possibilistic IM contour, but this turns out to be analogous to a block-specific calculation, as discussed below.  

Proposition \ref{p:gamma-IM} below gives the more-or-less closed-form expression of the block-specific possibility contour; we say ``more-or-less'' because it involves the Gamma and Lambert $W$ functions \citep{lambert1771observations,corless1996lambert}, which can be readily evaluated, e.g., using the {\tt gammainc} and {\tt lambertWp} functions in the {\tt R} packages {\tt expint} and {\tt lamW} respectively, but technically have no closed-form expression.  The proof highlights the complicated nature of performing the validification in closed-form. Among other difficulties, the mapping from observations $z^{(b)}$ to the likelihood ratio $R (z^{(b)}, \theta)$ is surjective, as in Example \ref{ss:gaussian}, which requires careful treatment in the derivation of the contour.

\begin{prop}\label{p:gamma-IM}
The IM possibility contour for the $b^\text{th}$ block is
\begin{align*}
\pi_{z^{(b)}}(\theta) &= \frac{1}{n_b!} \int_0^{t_b} \{ W_{0}(-n_b^{-1}t^{1/n_b}) + 1 \}^{-1} - \{ W_{-1}(-n_b^{-1} t^{1/n_b}) + 1 \}^{-1} ~dt\\
&=1 + \frac{\Gamma\{ n_b, -n_b W_{-1}(-n_b^{-1} t_b^{1/n_b}) \}}{(n_b-1)!} -\frac{\Gamma\{ n_b, -n_b W_{0}(-n_b^{-1} t_b^{1/n_b}) \}}{(n_b-1)!},
\end{align*}
with $t_b=(\theta n_b \hat{\theta}_{z^{(b)}}^{-1})^{n_b} e^{-\theta n_b \hat{\theta}_{z^{(b)}}^{-1}}$ and $W_{-1}, W_0$ two branches of the Lambert W function. 
\end{prop}

\begin{proof}
Define $X_b=\theta n_b \hat{\theta}_{Z^{(b)}}^{-1} \sim \gam(n_b,1)$ and $T_b= X_b^{n_b} e^{-X_b}$, and $t_b$ the observed value of $T_b$. The IM possibility contour for the $b$\textsuperscript{th} block is
\begin{align*}
\pi_{z^{(b)}}(\theta) &= \prob_{\theta} [( \theta \hat{\theta}_{Z^{(b)}}^{-1} )^{n_b} e^{n_b-\theta n_b \hat{\theta}_{Z^{(b)}}^{-1}} \leq ( \theta \hat{\theta}_{z^{(b)}}^{-1} )^{n_b} e^{n_b-\theta n_b \hat{\theta}_{z^{(b)}}^{-1}} ] =\prob_{\theta} (T_b \leq t_b).
\end{align*}
To evaluate the possibility contour, we find the density of $T_b= X_b^{n_b} e^{-X_b}$ using a change of variables. Let $t_b=x_b^{n_b} e^{-x_b}$, and the inverse transformation exists and corresponds to $x_b=-n_b W_{r_b}(-n_b^{-1} t_b^{1/n_b})$, with $W_{r_b}$ the Lambert W function, $r_b=-1$ if $n_b^{-1} x_b \geq 1$, i.e. if $\hat{\theta}_{z^{(b)}}^{-1} \geq \theta^{-1}$, otherwise $r_b=0$. The partial derivative is
\begin{align*}
\frac{\partial x_b}{\partial t_b} &=-\frac{W_{r_b}(-n_b^{-1} t_b^{1/n_b})}{t_b W_{r_b}(-n_b^{-1} t_b^{1/n_b}) + t_b} = \{n_b t_b^{1-1/n_b} e^{W_{r_b}(-n_b^{-1}t_b^{1/n_b})} -t_b\}^{-1}.
\end{align*}
Using separate change of variables on the two parts $\{x_b: n_b^{-1} x_b \geq 1\}$ and $\{ x_b: 0 \leq n_b^{-1} x_b <1\}$ of the domain of $X_b$, the distribution of $T_b$ is
\begin{align}
f_{T_b}(t_b) 
&= f_{X_b}(x_b) \Bigl| \frac{\partial x_b}{\partial t_b} \Bigr| \mathbbm{1}(n_b^{-1} x_b \geq 1) + f_{X_b}(x_b) \Bigl| \frac{\partial x_b}{\partial t_b} \Bigr| \mathbbm{1}(0 \leq n_b^{-1} x_b < 1), \nonumber \\
&=\frac{1}{\Gamma(n_b)} x_b^{n_b-1} e^{-x_b} \Bigl| \frac{\partial x_b}{\partial t_b} \Bigr| \mathbbm{1}(n_b^{-1} x_b \geq 1) + \frac{1}{\Gamma(n_b)} x_b^{n_b-1} e^{-x_b} \Bigl| \frac{\partial x_b}{\partial t_b} \Bigr| \mathbbm{1}(0 \leq n_b^{-1} x_b < 1) \nonumber \\
&= \frac{1}{(n_b-1)!} \Bigl\{ -n_b W_{-1}(-n_b^{-1}t_b^{1/n_b}) \Bigr\}^{n_b-1} e^{n_b W_{-1}(-n_b^{-1} t_b^{1/n_b})} \frac{W_{-1}(-n_b^{-1} t_b^{1/n_b})}{t_b W_{-1}(-n_b^{-1} t_b^{1/n_b}) + t_b} \nonumber \\
&\quad - \frac{1}{(n_b-1)!} \Bigl\{ -n_b W_{0}(-n_b^{-1} t_b^{1/n_b}) \Bigr\}^{n_b-1} e^{n_b W_{0}(-n_b^{-1} t_b^{1/n_b})} \frac{W_{0}(-n_b^{-1} t_b^{1/n_b})}{t_b W_{0}(-n_b^{-1} t_b^{1/n_b}) + t_b}, \label{eq:prop1-1}
\end{align}
with $0 \leq t_b \leq n_b^{n_b} e^{-n_b}$. Now, we use the identity $W_{r_b}(a)=ae^{-W_{r_b}(a)}$, $r_b \in \{-1,0\}$, of the Lambert W function to simplify
\begin{align*}
\frac{W_{r_b}(-n_b^{-1} t_b^{1/n_b})}{t_b W_{r_b}(-n_b^{-1} t_b^{1/n_b}) + t_b} &= \frac{-n_b^{-1}t_b^{1/n_b}e^{-W_{r_b}(-n_b^{-1}t_b^{1/n_b})}}{- t_b n_b^{-1} t_b^{1/n_b}e^{-W_{r_b}(-n_b^{-1}t_b^{1/n_b})} + t_b}.
\end{align*}
Multiplying the numerator and denominator by $-n_bt_b^{-1/n_b}e^{W_{r_b}(-n_b^{-1}t_b^{1/n_b})}$ gives
\begin{align*}
\frac{W_{r_b}(-n_b^{-1} t_b^{1/n_b})}{t_b W_{r_b}(-n_b^{-1} t_b^{1/n_b}) + t_b} 
&=\frac{1}{t_b - n_b t_b^{1-1/n_b} e^{W_{r_b}(-n_b^{-1}t_b^{1/n_b})}}.
\end{align*}
Plugging the above into equation \eqref{eq:prop1-1}, $f_{T_b}(t_b)$ simplifies to
\begin{align*}
f_{T_b}(t_b) &=\frac{1}{(n_b-1)!} \Bigl\{ -n_b W_{0}(-n_b^{-1} t_b^{1/n_b}) \Bigr\}^{n_b-1} \frac{e^{n_b W_{0}(-n_b^{-1} t_b^{1/n_b})} }{ n_b t_b^{1-1/n_b} e^{W_{0}(-n_b^{-1} t_b^{1/n_b})} -t_b }\\
&\quad -\frac{1}{(n_b-1)!} \Bigl\{ -n_b W_{-1}(-n_b^{-1}t_b^{1/n_b}) \Bigr\}^{n_b-1} \frac{e^{n_b W_{-1}(-n_b^{-1} t_b^{1/n_b})} }{n_b t_b^{1-1/n_b} e^{W_{-1}(-n_b^{-1} t_b^{1/n_b})} - t_b},
\end{align*}
with $0 \leq t_b \leq n_b^{n_b} e^{-n_b}$. Again using the identity $W_{r_b}(a)=ae^{-W_{r_b}(a)}$ to simplify the bracketed terms, we can write
\begin{align*}
f_{T_b}(t_b) &=\frac{1}{(n_b-1)!} \Bigl\{ t_b^{1/n_b} e^{-W_{0}(-n_b^{-1} t_b^{1/n_b})} \Bigr\}^{n_b-1} \frac{e^{n_b W_{0}(-n_b^{-1} t_b^{1/n_b})} }{ n_b t_b^{1-1/n_b} e^{W_{0}(-n_b^{-1} t_b^{1/n_b})} - t_b } \\
&\quad -\frac{1}{(n_b-1)!} \Bigl\{ t_b^{1/n_b} e^{-W_{-1}(-n_b^{-1} t_b^{1/n_b})} \Bigr\}^{n_b-1} \frac{e^{n_b W_{-1}(-n_b^{-1} t_b^{1/n_b})} }{n_b t_b^{1-1/n_b} e^{W_{-1}(-n_b^{-1} t_b^{1/n_b})} -t_b} \\
&=\frac{1}{(n_b-1)!}t_b^{1-1/n_b} e^{W_{0}(-n_b^{-1} t_b^{1/n_b}) - n_b W_{0}(-n_b^{-1} t_b^{1/n_b})} \frac{e^{n_b W_{0}(-n_b^{-1} t_b^{1/n_b})} }{ n_b t_b^{1-1/n_b} e^{W_{0}(-n_b^{-1} t_b^{1/n_b})} - t_b } \\
&\quad -\frac{1}{(n_b-1)!} t_b^{1-1/n_b} e^{W_{-1}(-n_b^{-1} t_b^{1/n_b})-n_b W_{-1}(-n_b^{-1} t_b^{1/n_b})} \frac{e^{n_b W_{-1}(-n_b^{-1} t_b^{1/n_b})} }{n_b t_b^{1-1/n_b} e^{W_{-1}(-n_b^{-1} t_b^{1/n_b})} - t_b},
\end{align*}
$0 \leq t_b \leq n_b^{n_b} e^{-n_b}$, where the second equality is obtained by expanding the brackets in the first. Carrying out routine simplifications, we get, for $0 \leq t_b \leq n_b^{n_b} e^{-n_b}$,
\begin{align*}
f_{T_b}(t_b) &=\frac{1}{(n_b-1)!} e^{W_{0}(-n_b^{-1} t_b^{1/n_b}) - n_b W_{0}(-n_b^{-1} t_b^{1/n_b})} \frac{e^{n_b W_{0}(-n_b^{-1} t_b^{1/n_b})} }{ n_b e^{W_{0}(-n_b^{-1} t_b^{1/n_b})} -t_b^{1/n_b} } \\
&\quad -\frac{1}{(n_b-1)!} e^{W_{-1}(-n_b^{-1} t_b^{1/n_b})-n_b W_{-1}(-n_b^{-1} t_b^{1/n_b})} \frac{e^{n_b W_{-1}(-n_b^{-1} t_b^{1/n_b})} }{n_b e^{W_{-1}(-n_b^{-1} t_b^{1/n_b})} -t_b^{1/n_b}} \\
&=\frac{1}{(n_b-1)!} \frac{e^{W_{0}(-n_b^{-1} t_b^{1/n_b})} }{ n_b e^{W_{0}(-n_b^{-1} t_b^{1/n_b})} -t_b^{1/n_b} } -\frac{1}{(n_b-1)!} \frac{e^{W_{-1}(-n_b^{-1} t_b^{1/n_b})} }{n_b e^{W_{-1}(-n_b^{-1} t_b^{1/n_b})} -t_b^{1/n_b}} \\
&=\frac{1}{(n_b-1)!}\{ n_b -t_b^{1/n_b} e^{-W_{0}(-n_b^{-1} t_b^{1/n_b})} \}^{-1} -\frac{1}{(n_b-1)!} \{ n_b -t_b^{1/n_b} e^{-W_{-1}(-n_b^{-1} t_b^{1/n_b})} \}^{-1}\\
&=\frac{1}{n_b!}\{ 1 - n_b^{-1} t_b^{1/n_b} e^{-W_{0}(-n_b^{-1} z^{1/n_b})} \}^{-1} -\frac{1}{n_b!} \{ 1 - n_b^{-1} t_b^{1/n_b} e^{-W_{-1}(-n_b^{-1} t_b^{1/n_b})} \}^{-1}.
\end{align*}
Finally, again using the identity $W_{r_b}(a)=ae^{-W_{r_b}(a)}$, the density of $T_b$ is
\begin{align*}
f_{T_b}(t_b) &=\frac{1}{n_b!}\{ W_{0}(-n_b^{-1} t_b^{1/n_b}) + 1 \}^{-1} - \frac{1}{n_b!}\{ W_{-1}(-n_b^{-1} t_b^{1/n_b}) + 1 \}^{-1}, \quad 0 \leq t_b \leq n_b^{n_b} e^{-n_b}.
\end{align*}
Our focus now turns to computing the distribution function from the density of $T_b$. Let $u=-n_b^{-1}t^{1/n_b}$ such that $t=(-n_b)^{n_b} u^{n_b}$ and $dt=n_b (-n_b)^{n_b}u^{n_b-1} du$. By a change of variables, the distribution function of $T_b$ is
\begin{align*}
G_{T_b}(t_b)&= \frac{1}{n_b!} \int_0^{t_b} \{ W_{0}(-n_b^{-1} t^{1/n_b}) + 1 \}^{-1} - \{ W_{-1}(-n_b^{-1} t^{1/n_b}) + 1 \}^{-1} ~dt \\
&=\frac{(-n_b)^{n_b}}{(n_b-1)!}\int_{-n_b^{-1} t_b^{1/n_b}}^0 u^{n_b-1} \{W_{-1}(u)+1\}^{-1} - u^{n_b-1} \{W_{0}(u)+1\}^{-1} ~du.
\end{align*}
Let $v=W_{r_b}(u)$ such that $u=ve^v$ and $du=e^v(1+v)dv$ for $r_b\in \{0,1\}$. Then by another change of variables,
\begin{align*}
G_{T_b}(t_b)
&=\frac{(-n_b)^{n_b}}{(n_b-1)!} \Biggl[ - \int_{-\infty}^{W_{-1}(-n_b^{-1} t_b^{1/n_b})} v^{n_b-1} e^{vn_b} ~dv - \int_{W_0(-n_b^{-1} t_b^{1/n_b})}^0 v^{n_b-1} e^{vn_b} ~dv\Biggr].
\end{align*}
By a third (and final) change of variables, letting $w=-vn_b$ such that $v=-w/n_b$ and $dv=-1/n_b dw$,
\begin{align*}
G_{T_b}(t_b)
&=\frac{1}{(n_b-1)!} \Biggl[ \int_0^{-n_bW_0(-n_b^{-1} t_b^{1/n_b})} w^{n_b-1} e^{-w} ~dw + \int_{-n_bW_{-1}(-n_b^{-1} t_b^{1/n_b})}^{\infty}  w^{n_b-1} e^{-w} ~dw \Biggr].
\end{align*}
Recognizing that the upper and lower incomplete Gamma functions are, respectively, the integrals $\Gamma(s,x)=\int_x^{\infty} w^{s-1} e^{-w}~dw$ and $\gamma(s,w)=\int_0^x w^{s-1}e^{-w}~dw$, the distribution function of $T_b$ is
\begin{align*}
G_{T_b}(t_b)
&= \frac{1}{(n_b-1)!} [ \gamma\{n_b, -n_bW_0(-n_b^{-1} t_b^{1/n_b})\} + \Gamma\{n_b, -n_bW_{-1}(-n_b^{-1} t_b^{1/n_b})\}].
\end{align*}
Finally, recall that
\begin{align*}
\gamma\{n_b,-n_bW_0(-n_b^{-1} t_b^{1/n_b})\} + \Gamma\{n_b,-n_bW_0(-n_b^{-1} t_b^{1/n_b})\}=\Gamma(n_b)=(n_b-1)!,
\end{align*}
and so
\begin{align*}
G_{T_b}(t_b) 
&=\frac{1}{(n_b-1)!} [(n_b-1)! - \Gamma\{n_b,-n_bW_0(-n_b^{-1} t_b^{1/n_b})\} + \Gamma\{n_b, -n_bW_{-1}(-n_b^{-1} t_b^{1/n_b})\}]\\ 
&=1 + \frac{1}{(n_b-1)!} \Gamma\{ n_b, -n_b W_{-1}(-n_b^{-1} t_b^{1/n_b}) \} -\frac{1}{(n_b-1)!}  \Gamma\{ n_b, -n_b W_{0}(-n_b^{-1}t_b^{1/n_b}) \}.
\end{align*}
Therefore, the individual IM's possibility contour is
\begin{align*}
\pi_{z^{(b)}}(\theta) &= \frac{1}{n_b!} \int_0^{t_b} \{ W_{0}(-n_b^{-1}t^{1/n_b}) + 1 \}^{-1} - \{ W_{-1}(-n_b^{-1} t^{1/n_b}) + 1 \}^{-1} ~dt\\
&=1 + \frac{1}{(n_b-1)!} \Gamma\{ n_b, -n_b W_{-1}(-n_b^{-1} t_b^{1/n_b}) \} -\frac{1}{(n_b-1)!}  \Gamma\{ n_b, -n_b W_{0}(-n_b^{-1} t_b^{1/n_b}) \},
\end{align*}
with $t_b=(\theta n_b \hat{\theta}_{z^{(b)}}^{-1})^{n_b} e^{-\theta n_b \hat{\theta}_{z^{(b)}}^{-1}}$.
\end{proof}

We can also derive the possibility contour for the full data from the summary statistics of the $B$ blocks. Using Proposition \ref{p:gamma-IM}, the full data IM possibility contour is
\begin{align*}
\pi_{z^n}(\theta) 
&=1 + \frac{\Gamma\{ n, -n W_{-1}(-n^{-1}t^{1/n}) \}}{(n-1)!} - \frac{\Gamma\{ n, -n W_{0}(-n^{-1}t^{1/n}) \}}{(n-1)!},
\end{align*}
with $t=(\theta\sum_{b=1}^B n_b \hat{\theta}_{z^{(b)}}^{-1})^n e^{-\theta\sum_{b=1}^B n_b \hat{\theta}_{z^{(b)}}^{-1}}$.

\section{Details from Section~\ref{s:asymptotic-rule}}
\label{a:large.sample}

\subsection{Regularity conditions and preliminary results}

\newcommand{\sL}{\mathcal{L}}

Certain regularity conditions are required in order to establish asymptotic concentration properties of estimators, posterior distributions, etc., and the same is true for IMs and the inner probabilistic approximations under consideration here.  Roughly, these conditions ensure that the log-likelihood function is smooth enough that it can be well-approximated by a quadratic function.  One common set of regularity conditions are the classical {\em Cram\'er conditions} \citep{cramer.book}, versions of which can be found in the standard texts, including \citet[][Theorem~3.10]{lehmann.casella.1998} and \citet[][Theorem~7.63]{schervish1995}.  Here we adopt the more modern set of sufficient conditions originating in \citet{lecam1956, lecam1960, lecam1970} and \citet{hajek1972}; see, also, \citet{bickel1998} and \citet{vaart1998}. 

The model specifies a class $\{\prob_\theta: \theta \in \TT\}$ of probability distributions, supported on $\ZZ$, indexed by $\TT \subseteq \RR^p$, with $\prob_\theta$ having a density $p_\theta(x)$ relative to a $\sigma$-finite measure $\nu$ on $\ZZ$.  Then the data $Z^n=(Z_1,\ldots,Z_n)$ is assumed to be independent and identically distributed of size $n$ with common distribution $\prob_\Theta$, where $\Theta \in \TT$ is the uncertain ``true value'' of the parameter.  Following \citet[][Ch.~2]{bickel1998}, define the (natural) logarithm and square-root density functions, respectively, as
\[ \ell_\theta(z) = \log p_\theta(z) \quad \text{and} \quad q_\theta(z) = p_\theta(z)^{1/2}. \]
The ``dot'' notation, e.g., $\dot g_\theta(z)$, represents a function that behaves like the derivative of $g_\theta(z)$ with respect to $\theta$ for pointwise in $z$.  If the usual partial derivative of $g_\theta(z)$ with respect to $\theta$ exists, then $\dot g_\theta(z)$ is that derivative; but suitable functions $\dot g_\theta(z)$ may exist even when the ordinary derivative fails to exist.  In particular, we assume existence of a suitable ``derivative'' $\theta \mapsto \dot q_\theta(z)$ of the square-root density.  Finally, let $\sL_2(\nu)$ denote the set of measurable functions on $\ZZ$ that are square $\nu$-integrable. 

\begin{asmp}
The parameter space $\TT$ is open and there exists a vector $\dot q_\theta(z) = \{ \dot q_{\theta,j}(z): j=1,\ldots,p\}$, whose coordinates $\dot q_{\theta,j}$ are elements of $\sL_2(\nu)$, such that the following conditions hold: 
\begin{itemize}
\item[A1.] the maps $\theta \mapsto \dot q_{\theta,j}$ from $\TT$ to $\sL_2(\nu)$ are continuous for each $j=1,\ldots,p$; 
\item[A2.] at each $\theta \in \TT$, 
\begin{equation}
\label{eq:dqm}
\int \bigl| q_{\theta + u}(z) - q_\theta(z) - u^\top \dot q_\theta(z) \bigr|^2 \, \nu(dz) = o(\|u\|^2), \quad u \to 0 \in \RR^p; 
\end{equation}
\item[A3.] and the $p \times p$ matrix $\int \dot q_\theta(z) \, \dot q_\theta(z)^\top \, \nu(dz)$ is non-singular for each $\theta \in \TT$. 
\end{itemize} 
\end{asmp}

The condition in \eqref{eq:dqm} is often described as $\theta \mapsto q_\theta$ being {\em differentiable in quadratic mean}.  Note that this condition does not require the square-root density to actually be differentiable at $\theta$, only that it be ``locally linear'' in a certain average sense.  The classical Cram\'er conditions assume more than two continuous derivatives, so \eqref{eq:dqm}, which does not even require existence of a first derivative, is significantly weaker; sufficient conditions for \eqref{eq:dqm} are given in \citet[][Lemma~7.6]{vaart1998}.  Then the {\em score function} $\dot\ell_\theta(z)$ is defined in terms of $\dot q_\theta(z)$ as 
\[ \dot\ell_\theta(z) = \frac{2\dot q_\theta(z)}{q_\theta(z)} \, 1\{q_\theta(z) > 0\}, \]
and it can be shown that $\int \dot\ell_\theta(z) \, \prob_\theta(dz) = 0$ for each $\theta$.  Moreover, Condition~A3 above implies non-singularity of the Fisher information matrix $I_\theta = \int \dot\ell_\theta(z) \ell_\theta(z)^\top \, \prob_\theta(dz)$ for each $\theta \in \TT$.  \citet[][Prop.~2.1.1]{bickel1998} provide sufficient conditions for A1--A3.   

The above conditions are rather mild, and hold in a broad range of problems, including exponential families and many more.  For some additional simplicity and interpretability, and with virtually no loss of generality, we will further assume that the log-likelihood function is twice continuously differentiable, which allows for alternative centering and scaling based on maximum likelihood estimators (rather than the true $\Theta$) and the observed Fisher information (rather than the Fisher information for a sample of size $1$).  

One further condition is required for the possibilistic Bernstein--von Mises theorem of \citet{martin2025asymptotic}, applied in Section~\ref{s:asymptotic-rule} above, namely, that the maximum likelihood estimator is consistent, i.e., $\hat\theta_{Z^n} \to \Theta$ in $\prob_\Theta$-probability as $n \to \infty$.  This, of course, is not automatic, but holds quite broadly when the dimension of the parameter space is fixed, as we are assuming here.  

A standard but relevant consequence of the conditions imposed above is the following: the full-data maximum likelihood estimator $\hat\theta_{Z^n}$ and the block weighted average $\check\theta_{S^n}$ are asymptotically equivalent.  

\begin{lem}\label{lemma:meta}
Under the stated conditions above, 
\begin{align*}
n^{1/2}(\hat\theta_n - \Theta) \to \nm_p(0, I_\Theta) \quad \text{and} \quad n^{1/2}(\check{\theta}_{S^n} - \Theta) \to \nm_p(0,I_\Theta),
\end{align*}
both in distribution under $\prob_{\Theta}$, where $I_\Theta$ is the $p \times p$ Fisher information matrix based on a sample of size 1.  Moreover, $n^{1/2}(\check\theta_{S^n} - \hat\theta_{Z^n}) \to 0$ in $\prob_{\Theta}$-probability. 
\end{lem}

\begin{proof}
The first two claims follow from the usual asymptotic Gaussianity of the maximum likelihood estimator under the usual regularity conditions. Write
\begin{align*}
n^{1/2}(\check{\theta}_{S^n} - \Theta) = n^{1/2}(\check{\theta}_{S^n} - \hat\theta_{Z^n}) + n^{1/2}(\hat\theta_{Z^n} - \Theta).
\end{align*}
The terms on the far right and far left of the above display have the same asymptotic distribution.  If $n^{1/2}(\check{\theta}_{S^n} - \hat\theta_{Z^n})$ converged to anything other than 0, then that would contradict the lemma's first statement. This implies the lemma's second claim. 
\end{proof}

\subsection{Proof of Theorem~\ref{thm:naive-efficiency}}

To start, we immediately get 
\begin{equation}
\label{eq:bound}
\begin{split}
\sup_{u \in C} \bigl| & \pi_{s^n}^\infty(\hat\theta_{Z^n} + J_{Z^n}^{-1/2} u) - \pi_{Z^n}(\hat\theta_{Z^n} + J_{Z^n}^{-1/2} u) \bigr| \\
& \leq \sup_{u \in C} \bigl| \pi_{s^n}^\infty(\hat\theta_{Z^n} + J_{Z^n}^{-1/2} u) - \gamma(u) \bigr| + \sup_{u \in C} \bigl| \pi_{Z^n}(\hat\theta_{Z^n} + J_{Z^n}^{-1/2} u) - \gamma(u) \bigr|, 
\end{split}
\end{equation}
where $\gamma(u) = 1 - F_p(\|u\|^2)$ is the standard $p$-dimensional Gaussian possibility contour, with $F_p$ the $\chisq(p)$ distribution function.  Then the main result in \cite{martin2025asymptotic} implies that the latter term is vanishing in $\prob_{\Theta}$-probability.  So, it is enough for us to focus on the former term in the upper bound \eqref{eq:bound}.  Define the quadratic form
\[ Q_{S^n}(u) = (\check{\theta}_{S^n} - \hat\theta_{Z^n} - J_{Z^n}^{-1/2} u)^\top J_{S^n} (\check{\theta}_{S^n} - \hat\theta_{Z^n} - J_{Z^n}^{-1/2} u), \]
so that the former term above can be re-expressed as 
\[ \bigl| \pi_{s^n}^\infty(\hat\theta_{Z^n} + J_{Z^n}^{-1/2} u) - \gamma(u) \bigr| = \bigl| F_p\{Q_{S^n}(u)\} - F_p(\|u\|^2) \bigr|. \]
By the mean-value theorem, 
\[ \bigl| F_p\{Q_{S^n}(u)\} - F_p(\|u\|^2) \bigr| \leq \bigl\{ \sup_r f_p(r) \bigr\} \times \bigl| Q_{S^n}(u) - \|u\|^2 \bigr|, \]
where $f_p$ is the derivative of $F_p$, i.e., the non-negative $\chisq(p)$ density function, and $r$ varies over the range of $Q_{S^n}(u)$ and/or $\|u\|^2$ for $u \in C$.  Since the leading factor in the upper bound is $O(1)$, it suffices to consider the difference $Q_{S^n}(u) - \|u\|^2$.  Clearly, 
\begin{align*}
\bigl| Q_{S^n}(u) - \|u\|^2 \bigr| & \leq (\check{\theta}_{S^n} - \hat\theta_{Z^n})^\top J_{S^n} (\check{\theta}_{S^n} - \hat\theta_{Z^n}) \\
& \qquad + 2 \bigl| (\check{\theta}_{S^n} - \hat\theta_{Z^n})^\top J_{S^n} J_{Z^n}^{-1/2} u \bigr| \\
& \qquad + \bigl| u^\top J_{Z^n}^{-1/2} J_{S^n} J_{Z^n}^{-1/2} u - \|u\|^2 \bigr|. 
\end{align*}
The first two terms are $o_P(1)$ by Lemma~\ref{lemma:meta} above and some details discussed below; the second term requires first an application of the Cauchy--Schwarz inequality and then is seen to be $o_P(1)$ uniformly for $u \in C$.  For the third term, we get 
\begin{align*}
\bigl| u^\top J_{Z^n}^{-1/2} J_{S^n} J_{Z^n}^{-1/2} u - \|u\|^2 \bigr| \leq \text{eig}_\text{max} \bigl\{ J_{Z^n}^{-1/2} J_{S^n} J_{Z^n}^{-1/2} - I_{p \times p}\bigr\} \, \|u\|^2,
\end{align*}
with $I_{p \times p}$ the identity matrix.  Note that $J_{Z^n} = n I_\Theta + o_P(n)$ and $J_{S^n} = n I_\Theta + o_P(n)$, where the latter claim follows since $J_{S^n} = \sum_{b=1}^B J_{Z^{(b)}}$, with $J_{Z^{(b)}} = n_b I_\Theta + o_P(n)$ by \eqref{eq:n.condition}.  This implies that the eigenvalue in the above display is $o_P(1)$, so the upper bound vanishes uniformly for $u \in C$.  Putting everything together, we get that both terms in the upper bound \eqref{eq:bound} are vanishing in $\prob_{\Theta}$-probability, which proves the claim.   

\section{Details from Section~\ref{s:proposal}}\label{a:proposal}

\subsection{Computation of confidence intervals}

The left and right panels of Figure~\ref{f:example-stat-inference} illustrate the construction of a $80\%$ confidence interval for $\Theta$ in the Gaussian and Exponential examples, respectively, of Section~\ref{ss:thule}. In the Gaussian example, $\pi_{s^n}^\forall(\theta)=\pi_{s^n}^\infty(\theta)=\pi_{z^n}(\theta)$. In the Exponential example, the valid divide-and-conquer contour $\pi_{s^n}^\forall(\theta)$ is estimated using Monte Carlo samples following our proposal in Section \ref{ss:middle-computation} (see equation \eqref{eq:contour.mc}) with $\theta^\dagger=\check{\theta}_{s^n}$ and $M=50{,}000$. 

\begin{figure}[t]
\centering
\begin{subfigure}[t]{0.49\textwidth}
\includegraphics[width=\textwidth]{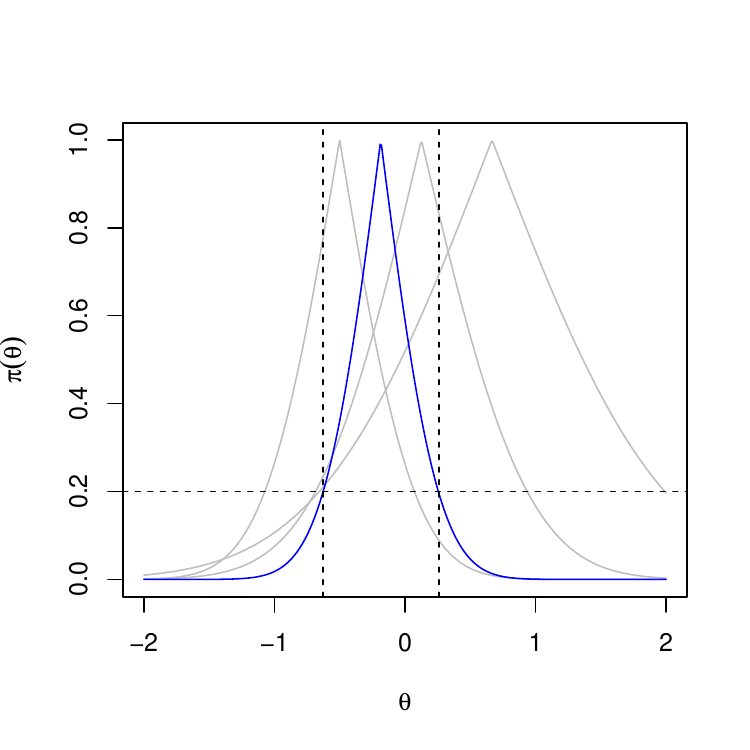}
\end{subfigure}
\begin{subfigure}[t]{0.49\textwidth}
\includegraphics[width=\textwidth]{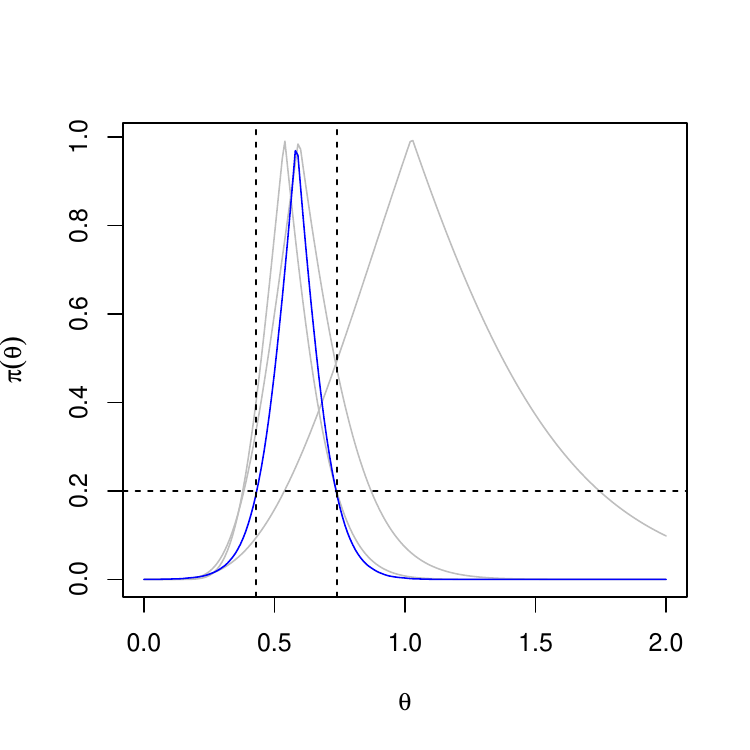}
\end{subfigure}
\caption{Individual (grey) and valid divide-and-conquer (blue) possibility contours in the Gaussian example (left) and the Exponential example (right). Horizontal and vertical dashed black lines illustrate the construction of a 80\% confidence interval for $\Theta$.} \label{f:example-stat-inference}
\end{figure}

\subsection{Proof of Theorem~\ref{thm:middle_efficiency}}

Introduce the notation 
\[ G_n^\theta(r) = \prob_\theta\{ R^\text{gauss}(S^n, \theta) \leq r \}, \quad r \in [0,1], \]
so that 
\[ \pi_{s^n}^\forall(\theta) = G_n^\theta\{ R^\text{gauss}(S^n, \theta) \}, \quad \theta \in \TT. \]
Following \cite{hedges1981distribution, hedges1983combining}, under the stated regularity conditions, 
\[ -2\log R^\text{gauss}(S^n, \theta) \to \chisq(p) \quad \text{in distribution under $\prob_\theta$, for each $\theta \in \TT$}. \]
By this and the continuous mapping theorem, we know that 
\[ | G_n^\theta(r) - G(r) | \to 0, \quad \text{for each $r$ as $n \to \infty$ and each $\theta$}, \]
where $G$ is the distribution function of a $\exp\{ -\frac12 \chisq(p) \}$ random variable.  Since these are distribution functions---bounded and non-decreasing---the pointwise convergence implies uniform convergence, i.e., 
\[ \| G_n^\theta - G \|_\infty := \sup_{r \in [0,1]} | G_n^\theta(r) - G(r) | \to 0, \quad \text{for each $\theta$}. \]
Importantly, we will show below that this holds (locally) uniformly in $\theta$.  

Like in the proof of Theorem \ref{thm:naive-efficiency}, we start with the bound 
\begin{align*}
\sup_{u \in C} \bigl| \pi_{s^n}^\forall(\hat\theta_{Z^n}& + J_{Z^n}^{-1/2}u) - \pi_{Z^n}(u) \bigr| \\
& \leq \sup_{u \in C} \bigl| \pi_{s^n}^\forall(\hat\theta_{Z^n} + J_{Z^n}^{-1/2}u) - \gamma(u) \bigr| + \sup_{u \in C} \bigl| \pi_{Z^n}(\hat\theta_{Z^n} + J_{Z^n}^{-1/2}u) - \gamma(u) \bigr|, 
\end{align*}
where the $\gamma(u) = G(e^{-\frac12 \|u\|^2})$ is the Gaussian possibility contour; this can also be expressed in terms of a $\chisq(p)$ distribution function, as we had done previously.  The second term in the above vanishes in $\prob_{\Theta}$-probability by \cite{martin2025asymptotic}, so we focus here on the first term.  We proceed with the further upper bound
\begin{align*}
\bigl| \pi_{s^n}^\forall&(\hat\theta_{Z^n} + J_{Z^n}^{-1/2}u) - \gamma(u) \bigr| \\
& \leq \bigl| G_n^{\vartheta_n^u}\{ R^\text{gauss}(S_n, \vartheta_n^u) \} - G\{ R^\text{gauss}(S^n, \vartheta_n^u) \} \bigr| + \bigl| G\{ R^\text{gauss}(S^n, \vartheta_n^u) \} - \gamma(u) \bigr|, 
\end{align*}
where $\vartheta_n^u = \check\theta_{S^n} + J_{S^n}^{-1/2}u$.  By definition of $R^\text{gauss}$, we have 
\[ R^\text{gauss}(S^n, \vartheta_n^u) \equiv e^{-\frac12 \|u\|^2}, \]
so the second term in the penultimate display is 0 and drops out completely.  Now, since the second term is gone and the first term involves a difference of two functions applied to the same argument, we get the bound:   
\begin{equation}
\label{eq:eq.bound}
\bigl| \check\pi_{s^n}^\forall(\hat\theta_{Z^n} + J_{Z^n}^{-1/2} u) - \gamma(u) \bigr| \leq \bigl\| G_n^{\vartheta_n^u} - G \bigr\|_\infty. 
\end{equation}
This connects our present goal with the (uniform) convergence of the $G_n^\theta$ distribution functions as mentioned briefly above. 

To simplify the notation, let $\vartheta_n$ be a (deterministic) sequence that is bounded within a compact set $\mathcal{T}$.  Then $G_n^{\vartheta_n}$ determines the distribution of the quadratic form 
\[ (\check\theta_{S^n} - \vartheta_n)^\top J_{S^n} \, (\check\theta_{S^n} - \vartheta_n), \]
when $\vartheta_n$ is {\em the true value of the parameter}.  The Wilks-like theorems do not impose any restrictions on the true parameter beyond that it is not on the boundary of the (open) parameter space and that the model is sufficiently smooth there.  Our regularity conditions already imply that smoothness holds everywhere---hence at $\vartheta_n$---and the restriction to a compact $\mathcal{T}$ keeps $\vartheta_n$ away from the boundary.  Therefore, 
\[ \bigl\| G_n^{\vartheta_n} - G \bigr\|_\infty \to 0 \quad \text{for every sequence $(\vartheta_n) \subset \mathcal{T}$}. \]
But if it holds for every sequence, then it must hold for the maximizer 
\[ \vartheta_n^\mathcal{T} = \arg\max_{\theta \in \mathcal{T}} \| G_n^\theta - G \|_\infty, \]
which, in turn, implies that 
\[ \sup_{\theta \in \mathcal{T}} \| G_n^\theta - G \|_\infty \to 0. \]
In our present case, the upper bound \eqref{eq:eq.bound} involves the collection of sequences of random variables $\vartheta_n^u$, indexed by $u$.  The asymptotic properties of $\check\theta_{S^n}$, as established in, e.g., \citet{hedges1981distribution}, ensure that $\|\vartheta_n^u - \Theta \| = O_P(n^{1/2})$ under $\prob_\Theta$, uniformly for $u \in C$.  So, for any compact $C$, there exists a compact $\mathcal{T}$ such that $\vartheta_n^u$ is in $\mathcal{T}$ for all $u \in C$ with $\prob_{\Theta}$-probability converging to 1.  Therefore, on the event where $\vartheta_n^u$ is constrained to the previously-defined $\mathcal{T}$, which depends on $C$, we have that 
\[ \sup_{u \in C} \bigl| \pi_{s^n}^\forall(\hat\theta_{Z^n} + J_{Z^n}^{-1/2} u) - \gamma(u) \bigr| \leq \bigl\| G_n^{\vartheta_n^u} - G \bigr\|_\infty \leq \sup_{\theta \in \mathcal{T}} \| G_n^\theta - G \|_\infty. \]
Since the upper bound in the above display is vanishing, and the aforementioned event has probability converging to 1, we conclude that the left-hand side of \eqref{eq:eq.bound} is vanishing in $\prob_{\Theta}$-probability, proving the claim. 

\subsection{Computation of valid divide-and-conquer IM} \label{a:importance-sampling}

In cases where there may be concerns about how close the working relative likelihood is to being an approximate pivot, there are alternatives.  A middle-ground between the na\"{i}ve strategy at the beginning of this subsection, and the strategy that fully leans into the ``approximately pivotal'' feature of the working relative likelihood in equation \eqref{eq:contour.mc} uses importance sampling.  Let $f_\theta(s^n)$ denote the density/mass function corresponding to the sampling distribution $\prob_{\theta}^{S^n}$ of the study-specific summary statistics under the posited model. Then we can approximate the valid divide-and-conquer contour by sampling $S_m^n = (S_1^m,\ldots,S_k^m)$ independently from $\prob_{\theta^\dagger}^{S^n}$ 
\begin{align}
\pi_{s^n}^\text{v}(\theta) \approx \frac{1}{M} \sum_{m=1}^M \mathbbm{1} \{ R_{S_m^n}^\text{work}(\theta) \leq R_{s^n}^\text{work}(\theta)\} \, w_m(\theta, \theta^\dagger), \quad \theta \in \TT,
\label{e:proposal-fused-IM-MC}
\end{align}
where the (importance) weights are the likelihood ratios 
\begin{align*}
w_m(\theta, \theta^\dagger) = \frac{ f_{\theta}(S_m^n)}{f_{\theta^\dagger}(S_m^n)}, \quad m=1,\ldots,M.
\end{align*}
The approximation ``$\approx$'' means that the right-hand side is an unbiased estimator of the left-hand side, so the former converges almost surely to the latter as $M \to \infty$.  Aside from the weights, there is another important difference between \eqref{eq:contour.mc} and \eqref{e:proposal-fused-IM-MC}, namely, the argument at which the working relative likelihood is evaluated: the former evaluates the working relative likelihood at the anchor $\theta^\dagger$ whereas the latter evaluates at the target value $\theta$.  This difference is important because evaluating the working relative likelihood at a parameter value different from that used to generate the summary statistics $S_1,\ldots,S_k$ will break whatever approximate pivot structure is present.  But the motivation for using this middle-ground computational strategy was that there may be doubts about how strong that approximate pivot structure is, so breaking would not be such a severe action to take.  The advantage to this approach is that one does not need to take different Monte Carlo samples for each $\theta$, as the original/na\"{i}ve Monte Carlo approach requires.  A similar-quality approximation should be possible with a relatively small subset of anchors at which the Monte Carlo samples can be taken, but it would depend on various aspects of the problem whether this safer, importance sampling-based approach is better than an approach that leans into the approximate pivot structure.  For example, in exact pivot cases like the Gaussian and exponential models, the importance sampling strategy requires many anchors whereas we know that the simple strategy of equation \eqref{eq:contour.mc} works perfectly with one arbitrarily chosen anchor.  

\section{Details from Section \ref{s:simulations}}

\subsection{Sampling from alpha-stable distributions} \label{a:alpha-stable}

The algorithm for sampling from alpha-stable distributions given here is due to \cite{chambers01061976}. Define $\zeta = -\beta \tan(\pi \alpha/2)$ and $\xi=\arctan (-\zeta)/\alpha$ if $\alpha \neq 1$ and $\xi=\pi/2$ if $\alpha=1$. First, sample two independent random variables $U \sim \unif(-\pi/2, \pi/2)$ and $W\sim \expo(1)$. Then the random variable
\begin{align*}
X &= \begin{cases}
(1+\zeta^2)^{\frac{1}{2\alpha}} \frac{\sin\{\alpha (U+\xi)\}}{\{ \cos(U)\}^{\frac{1}{\alpha}}} \bigl[ \frac{\cos\{ U-\alpha(U+\xi)\}}{W} \bigr]^{\frac{1-\alpha}{\alpha}} &\mbox{ if } \alpha \neq 1\\
\frac{1}{\xi} \bigl[ \bigl( \frac{\pi}{2} + \beta U \bigr) \tan(U) - \beta \log \bigl\{ \frac{\frac{\pi}{2} W \cos(U)}{\frac{\pi}{2} + \beta U} \bigr\} \bigr] &\mbox{ if } \alpha = 1
\end{cases}
\end{align*}
is alpha-stable with parameters $\alpha$, $\beta$, $c=1$ and $\mu=0$. From $X$ we can construct the random variable
\begin{align*}
Y &= \begin{cases}
cX + \mu & \mbox{ if } \alpha \neq 1\\
cX + \frac{2}{\pi} \beta c \log(c) + \mu & \mbox{ if } \alpha=1
\end{cases},
\end{align*}
which is alpha-stable with parameters $\alpha, \beta, c, \mu$.

\subsection{Full maximum likelihood for alpha-stable simulation} \label{a:alpha-stable-MLE}

To motivate the need for the validification step of the maximum likelihood inference based on the full data, described in Section \ref{ss:oracle-IM}, we compute large-sample confidence intervals for $\Theta$ in the $1{,}000$ replicates of Section \ref{ss:alpha-stable-sims}. Incidentally, of the $1{,}000$ replicates, $27$ have an observed information matrix that is not positive definite, illustrating the difficulty of numerical optimization. For each of the remaining $973$ replicates, we compute the $100(1-\alpha)\%$ confidence intervals for $\Theta$ using $\hat{\theta}_{z^{(b)}} \pm z_{\alpha/2} J_{z^{(b)}}^{-1/2}$. The marginal empirical coverage probability at level $100(1-\alpha)\%$, resported in Table \ref{t:alpha-stable-MLE-CP}, is the proportion of the $973$ computed intervals that contain the true value, $\Theta$. As there are no validity guarantees, the empirical coverage probabilities are below the nominal levels for some values of $\alpha$, in particular for $\beta$. 

\begin{table}[h!]
\caption{Empirical coverage probability (in \%) based on the full maximum likelihood estimator of the alpha-stable simulations.} \label{t:alpha-stable-MLE-CP}
\centering
\ra{0.8}
\begin{tabular}{rrrrrrrrrr}
Parameter & \multicolumn{9}{c}{$100(1-\alpha)\%$} \\
& $10$ & $20$ & $30$ & $40$ & $50$ & $60$ & $70$ & $80$ & $90$\\
\midrule
$\mu$ & 10.4 & 18.8 & 27.1 & 34.6 & 41.7 & 50.1 & 59.2 & 74.3 & 87.9 \\
$c$ & 8.32 & 19.2 & 30.4 & 40.7 & 50.7 & 60.7 & 70.9 & 80.1 & 89.3 \\
$\beta$ & 7.91 & 18.5 & 26.5 & 36.6 & 46.4 & 59.2 & 70.4 & 79.4 & 89.2 \\
\end{tabular}
\end{table}

\subsection{Full maximum likelihood for $g$-and-$k$ simulation}
\label{a:g-and-k-MLE}

To compare our divide-and-conquer approach to the maximum likelihood inference based on the full data, described in Section \ref{ss:oracle-IM}, we compute large-sample confidence intervals for $\Theta$ in the $1{,}000$ replicates of Section \ref{ss:alpha-stable-sims}. Of the $1{,}000$ replicates, $11$ have an observed information matrix (i.e. negative hessian of the log-likelihood evaluated at $\hat{\theta}_n$) that is not positive definite, illustrating the difficulty of numerical optimization. For each of the remaining $989$ replicates, we compute the $100(1-\alpha)\%$ confidence intervals for $\Theta$ using $\hat{\theta}_n \pm z_{\alpha/2} J_{n_i}^{-1/2}$. The marginal empirical coverage probability at level $100(1-\alpha)\%$, resported in Table \ref{t:g-and-k-MLE-CP}, is the proportion of the $989$ computed intervals that contain the true value, $\Theta$. While there are no validity guarantees, the empirical coverage probabilities seem to track across the nominal levels, seeming to suggest that the total sample size $n=200$ is large enough to give approximately valid inference. Of course, there is no guarantee that this will be the case, and so a comprehensive simulation for the full MLE would need to be carried out in each new setting.

\begin{table}[h!]
\centering
\caption{Empirical coverage probability (in \%) for full maximum likelihood inference in the $g$-and-$k$ simulations.} \label{t:g-and-k-MLE-CP}
\ra{0.8}
\begin{tabular}{rrrrrrrrrr}
Parameter & \multicolumn{9}{c}{$100(1-\alpha)\%$} \\
& $10$ & $20$ & $30$ & $40$ & $50$ & $60$ & $70$ & $80$ & $90$\\
\midrule
$\mu$ & 10.5 & 22.0 & 32.7 & 43.7 & 53.3 & 63.6 & 73.6 & 83.4 & 90.9 \\
$\sigma$ & 12.1 & 23.4 & 33.9 & 42.3 & 52.1 & 63.6 & 72.7 & 82.5 & 91.5 \\
$g$ & 8.09 & 17.8 & 28.0 & 38.8 & 49.5 & 59.6 & 69.7 & 80.8 & 89.1 \\
$k$ & 10.1 & 20.0 & 31.4 & 41.1 & 52.5 & 63.2 & 73.1 & 82.4 & 91.2 \\
\end{tabular}
\end{table}

\subsection{PM\textsubscript{2.5} analysis details} \label{a:data}

Figure \ref{f:b-spline} plots the splines used for the scale parameter of the $g$-and-$k$ distribution, omitting the intercept spline $B_{0,5}$.

\begin{figure}[h!]
\centering
\includegraphics[width=0.5\textwidth]{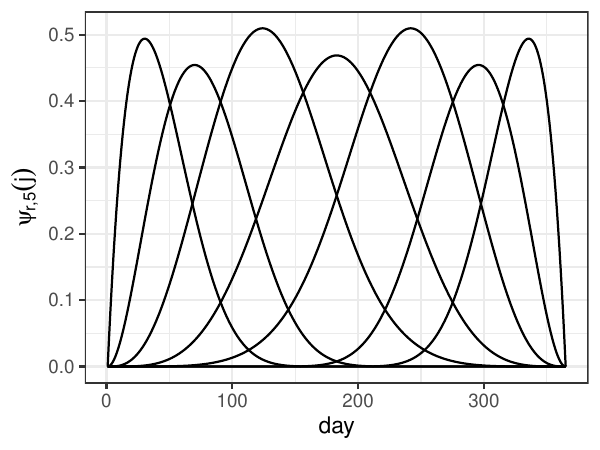}
\caption{B-spline basis for PM\textsubscript{2.5} data analysis.} \label{f:b-spline}
\end{figure}

\bibliographystyle{apalike}
\bibliography{IM_fusion-bib}

\end{document}